\documentclass{article}

\usepackage{amsmath}
\usepackage{amsfonts}
\usepackage{amssymb}
\usepackage{amsthm}

\usepackage[linesnumbered,boxed,noend]{algorithm2e}
\usepackage[version=latest]{pgf}
\usepackage{tikz}
\usetikzlibrary{arrows}

\usepackage{multirow}

\usepackage{tkz-graph}
\usepackage{proof}
\DeclareMathAlphabet{\mathpzc}{OT1}{pzc}{m}{it}

\usepackage{listings}

\usepackage{url}

\begin{document}
%
\title{Safety Verification of Phaser Programs}

\author{Zeinab Ganjei,
Ahmed Rezine,
Petru Eles and 
Zebo Peng\\
Link\"{o}ping University, Sweden
}

\date{}

\maketitle

\newtheorem{definition}{Definition}
\newtheorem{lemma}{Lemma}
\newtheorem{theorem}{Theorem}


\newcommand{\nats}{\ensuremath{\mathbb{N}}}
\newcommand{\ints}{\ensuremath{\mathbb{Z}}}

\newcommand{\infNats}{\ensuremath{\mathbb{N}^\infty}}

\newcommand{\set}[1]{\left\{{#1}\right\}}
\newcommand{\setcomp}[2]{\set{#1 ~|~#2 }}
\newcommand{\tuple}[1]{\left({#1}\right)}

\newcommand{\smset}[1]{\{{#1}\}}
\newcommand{\smsetcomp}[2]{\smset{#1 ~|~#2 }}
\newcommand{\smtuple}[1]{({#1})}

\newcommand{\upTo}[1]{[1,{#1}]}
\newcommand{\intSet}[2]{[{#1},{#2}]}

\newcommand{\sizeOf}[1]{|{#1}|}
\newcommand{\subsetsOf}[1]{2^{#1}}
\newcommand{\pFunctionsOf}[2]{{\mathtt{Pfn}}\tuple{{#1},{#2}}}

\newcommand{\tFunctionsOf}[2]{{\mathtt{Fn}}\tuple{{#1},{#2}}}

\newcommand{\undefined}{\uparrow}
\newcommand{\emptyFunction}[2]{\varnothing_{#1}^{#2}}

\newcommand{\defined}{\downarrow}

\newcommand{\subst}[3]{{#1}{\left[{#2}/{#3}\right]}}
\newcommand{\substSet}[2]{{#1}{\left[{#2}\right]}}
\newcommand{\mapSubst}[3]{{#1}{\left[{#2}\leftarrow\left({#3}\right)\right]}}

\newcommand{\restrict}[2]{{#1}_{|{#2}}}

\newcommand{\xSet}{X}
\newcommand{\ySet}{Y}
\newcommand{\kconsts}{K}
\newcommand{\xvar}{x}
\newcommand{\yvar}{y}
\newcommand{\ael}{a}
\newcommand{\bel}{b}
\newcommand{\zerovar}{zero}
\newcommand{\kconst}{k}
\newcommand{\pkconst}{k'}
\newcommand{\consts}{Const}
\newcommand{\vals}{Val}
\newcommand{\val}{val}
\newcommand{\gclause}{\delta}
\newcommand{\gcstr}{\Delta}
\newcommand{\ggraph}{\wp}
\newcommand{\topOf}[1]{\mathtt{topOf}\left(#1\right)}
\newcommand{\emptyGraph}{\ggraph_\false}
\newcommand{\ggraphs}{\mathcal{G}}
\newcommand{\oggraph}{\ggraph'}
\newcommand{\vectorOf}[1]{\mathtt{vectorOf}\left(#1\right)}
\newcommand{\ggraphOf}[1]{\mathtt{graphOf}\left(#1\right)}
\newcommand{\ggraphTuple}{\tuple{\vertices,\edges}}
\renewcommand{\iff}{\Longleftrightarrow}
\newcommand{\cnstr}{\mathcal{c}}
\newcommand{\vertices}{V}
\newcommand{\verticesOf}[1]{\mathtt{varsOf}(#1)}
\newcommand{\vertex}{v}
\newcommand{\nodeOf}[1]{\vertex_{#1}}
\newcommand{\edges}{E}
\newcommand{\proj}{Proj}
\newcommand{\merge}{\oplus}
\newcommand{\entails}{\sqsubseteq }
\newcommand{\comp}{\circ}
\newcommand{\closureOf}[1]{\mathtt{clo}\tuple{#1}}
\newcommand{\gedge}[3]{{#1}\xrightarrow{#2}{#3}}
\newcommand{\degreeOf}[1]{\mathtt{degreeOf}({#1})}
\newcommand{\addVars}{\oplus}
\newcommand{\removeVars}{\ominus}
\newcommand{\projectAway}[2]{{#1}_{#2\uparrow}}
\newcommand{\entailed}{\sqsubseteq_\ggraphs}
\newcommand{\satOf}[1]{Sat(#1)}
\newcommand{\isSat}[1]{\mathtt{isSat}(#1)}
\newcommand{\isRegistred}[3]{\mathtt{isReg}(#1,#2,#3)}
\newcommand{\registred}[2]{\mathtt{Reg}(#1,#2)}


\newcommand{\thVarSet}{\mathtt{T}}
\newcommand{\bVarSet}{\mathtt{B}}
\newcommand{\phVarSet}{\mathtt{V}}

\newcommand{\phvar}{\mathtt{v}}
\newcommand{\ophvar}{\mathtt{w}}
\newcommand{\thvar}{{\mathtt{thvar}}}
\newcommand{\bvar}{\mathtt{b}}
\newcommand{\var}{\mathtt{var}}

\newcommand{\seqVal}{\mathtt{s}}

\newcommand{\seqSet}{\mathtt{S}}
\newcommand{\dknow}{\mathtt{*}}

\newcommand{\grammarOr}{\;\mbox{\makebox[0mm][l]{\raisebox{-0.6ex}{\textbar}}%
    \raisebox{0.6ex}{\textbar}}\;}
\newcommand{\program}{\mathtt{prg}}
\newcommand{\programTuple}{\tuple{\bVarSet,\phVarSet,\tasks}}
\newcommand{\prgs}{\mathtt{Pr}}
\newcommand{\stmt}{\mathtt{stmt}}
\newcommand{\newt}{\mathtt{new\_thread()}}
\newcommand{\fork}[2]{\mathtt{fork(#1)\set{#2}}}
\newcommand{\newp}{\mathtt{newPhaser()}}
\newcommand{\reg}[2]{#1\mathtt{.reg(#2)}}
\newcommand{\asynch}[2]{\mathtt{asynch(#1,#2)}}
\newcommand{\dereg}[1]{#1\mathtt{.drop()}}
\newcommand{\sig}[1]{#1.\mathtt{signal()}}
\newcommand{\wait}[1]{#1.\mathtt{wait()}}
\newcommand{\while}[2]{\mathtt{while(#1)\set{ #2 }}}
\newcommand{\close}{\mathtt{end}}
\newcommand{\skipp}{\mathtt{skip}}
\newcommand{\cond}{\mathtt{cond}}
\newcommand{\bval}{\mathtt{b}}
\newcommand{\true}{\mathtt{true}}
\newcommand{\false}{\mathtt{false}}
\newcommand{\ndet}{\mathtt{ndet()}}
\newcommand{\ifp}[2]{\mathtt{if(#1) \set{  #2 }}}
\newcommand{\assert}[1]{\mathtt{assert(#1)}}
\newcommand{\exit}{\mathtt{exit}}
\newcommand{\task}{\mathtt{task}}
\newcommand{\tasks}{\mathtt{T}}
\newcommand{\parametersOf}[1]{\mathtt{paramOf}(#1)}


\newcommand{\thNum}{{t^{\bullet}}}
\newcommand{\phNum}{{p^{\bullet}}}
\newcommand{\valMp}{\eta}
\newcommand{\boolMp}{\pmb{bv}}
\newcommand{\pcMp}{\pmb{pc}}
\newcommand{\varMp}{\pmb{pv}}
\newcommand{\regMp}{\mathbf{reg}}
\newcommand{\cPhaseMp}{\pmb{\varphi}}
\newcommand{\sPhaseMp}{\pmb{\gamma}}
\newcommand{\mostGeneralsPhaseMps}[2]{\Gamma(#1,#2)}

\newcommand{\cConf}{c}
\newcommand{\cConfInit}{\cConf_{init}}
\newcommand{\sConf}{\phi}
\newcommand{\osConf}{\psi}
\newcommand{\sConfInit}{\phi_{init}}
\newcommand{\sConfBad}{\phi_{bad}}
\newcommand{\sConfSet}{\Phi}
\newcommand{\osConfSet}{\Psi}
\newcommand{\sConfSetInit}{\Phi_{init}}
\newcommand{\sConfSetBad}{\Phi_{bad}}

\newcommand{\relaxDegree}[2]{\rho_{#2}(#1)}

\newcommand{\errorConf}{\cConf_{\mathtt{err}}}
\newcommand{\errorConfAssert}{\mathtt{assertionfault}}
\newcommand{\cConfSetAssert}[2]{\mathtt{badConfs}_{\mathtt{assert}}^{(#1,#2)}}
\newcommand{\cConfSetRace}[2]{\mathtt{badConfs}_{\mathtt{race}}^{(#1,#2)}}
\newcommand{\cConfSetRuntime}[2]{\mathtt{badConfs}_{\mathtt{runtime}}^{(#1,#2)}}
\newcommand{\cConfSetDeadlock}[2]{\mathtt{badConfs}_{\mathtt{deadlock}}^{(#1,#2)}}
\newcommand{\controlRestrictAssert}[2]{\mathcal{C}_{\mathtt{assert}}^{(#1,#2)}}
\newcommand{\controlRestrictRace}[2]{\mathcal{C}_{\mathtt{race}}^{(#1,#2)}}
\newcommand{\controlRestrictRuntime}[2]{\mathcal{C}_{\mathtt{runtime}}^{(#1,#2)}}
\newcommand{\controlRestrictDeadlock}[2]{\mathcal{C}_{\mathtt{deadlock}}^{(#1,#2)}}
  \newcommand{\errorAssertConf}{\mathtt{err}_\mathtt{assert}}
      \newcommand{\errorRuntimeConf}{\mathtt{err}_\mathtt{run}}
      \newcommand{\raceConf}{\mathtt{err}_\mathtt{race}}
\newcommand{\sErrorConf}{\sConf_{\mathtt{err}}}
\newcommand{\sErrorConfAssert}{\sConf_{\mathtt{assertErr}}}
\newcommand{\sConfDeadlock}{\sConf_d}
\newcommand{\cConfSetD}[2]{\mathcal{D}^{#1}_{#2}}
\newcommand{\cConfSet}{\mathcal{C}}

\newcommand{\sConfSetAssert}[2]{\mathtt{badCstrs}_{\mathtt{assert}}^{(#1,#2)}}
\newcommand{\sConfSetRace}[2]{\mathtt{badCstrs}_{\mathtt{race}}^{(#1,#2)}}
\newcommand{\sConfSetRuntime}[2]{\mathtt{badCstrs}_{\mathtt{runtime}}^{(#1,#2)}}
\newcommand{\sConfSetDeadlock}[2]{\mathtt{badCstrs}_{\mathtt{deadlock}}^{(#1,#2)}}

\newcommand{\constraintOf}[1]{\mathtt{constraintOf}(#1)}

\newcommand{\entailedBy}{\sqsubseteq}

\newcommand{\cConfTuple}{\tuple{\thidSet,\phidSet,\boolMp,\pcMp,\varMp,\cPhaseMp}}

\newcommand{\sConfTuple}{\tuple{\thidSet,\phidSet,\boolMp,\pcMp,\varMp,\sPhaseMp}}

\newcommand{\denotationOf}[1]{[\![{#1}]\!]}

\newcommand{\thid}{t}
\newcommand{\othid}{u}
\newcommand{\oothid}{v}
\newcommand{\ooothid}{w}
\newcommand{\thidSet}{\mathpzc{T}}
\newcommand{\othidSet}{\mathpzc{U}}

\newcommand{\phid}{p}
\newcommand{\ophid}{q}
\newcommand{\oophid}{r}
\newcommand{\phidSet}{\mathpzc{P}}
\newcommand{\ophidSet}{\mathpzc{Q}}

\newcommand{\waitVal}[2]{wait_{#2}^{#1}}
\newcommand{\sigVal}[2]{sig_{#2}^{#1}}

\newcommand{\waitVar}[2]{\omega_{#2}^{#1}}
\newcommand{\sigVar}[2]{\sigma_{#2}^{#1}}

\newcommand{\vars}[2]{\chi_{#2}^{#1}}

\newcommand{\minD}[3]{min_{#3}^{{#1},{#2}}}
\newcommand{\maxD}[3]{max_{#3}^{{#1},{#2}}}

\newcommand{\reducesTo}[2]{\xrightarrow[#2]{#1}}

\newcommand{\preReducesTo}[1]{\stackrel{#1}{\rightsquigarrow}}


\newcommand{\Nodes}[1]{\mathpzc{N}_{#1}}
\newcommand{\taskNodes}{\Nodes{\thidSet}}
\newcommand{\ptaskNodes}{\Nodes{\thidSet'}}
\newcommand{\otaskNodes}{\Nodes{\othidSet}}
\newcommand{\phaserNodes}{\Nodes{\phidSet}}
\newcommand{\pphaserNodes}{\Nodes{\phidSet'}}
\newcommand{\ophaserNodes}{\Nodes{\ophidSet}}
\newcommand{\varNodes}[1]{\Nodes{#1}}

\newcommand{\undefPhaser}{\Node{\mathpzc{\bot}}}
\newcommand{\ephaserNodes}{{\phaserNodes^{\bot}}}

\newcommand{\Node}[1]{\mathpzc{n}_{#1}}
\newcommand{\oNode}[1]{\mathpzc{m}_{#1}}
\newcommand{\taskNode}{\mathpzc{t}}
\newcommand{\phaserNode}{\mathpzc{p}}
\newcommand{\varNode}{\mathpzc{v}}

\newcommand{\encodingOf}[1]{\mathtt{enc}(#1)}
\newcommand{\gEncodingOf}[1]{\mathtt{genc}(#1)}

\newcommand{\sgraph}{\mathpzc{sg}}

\newcommand{\gEdges}{\mathpzc{E}}

\newcommand{\gLabels}{\mathpzc{\lambda}}

\newcommand{\sLabels}{\mathpzc{\eta}}

\newcommand{\gTLabels}{\gLabels_{\thidSet}}
\newcommand{\gVLabels}{\gLabels_{V}}

\newcommand{\cGraph}{\mathpzc{cg}}

\newcommand{\taskBijection}{\tau}
\newcommand{\phaserBijection}{\pi}

\newcommand{\taskMapping}{\tau}
\newcommand{\phaserMapping}{\pi}

\newcommand{\sgEdges}{\gEdges}
\newcommand{\sGraph}{\mathpzc{sg}}
\newcommand{\sgLabels}{\gamma}
\newcommand{\sgTLabels}{\sgLabels_{\thidSet}}
\newcommand{\sgVLabels}{\sgLabels_{V}}
\newcommand{\sgPLabels}{\sgLabels_{P}}

\newcommand{\multiset}{m}

  \newcommand{\phaseOf}[1]{\mathtt{phaseOf}(#1)}


\newcommand{\sConfSetWorking}{\mathtt{Working}}
\newcommand{\sConfSetVisited}{\mathtt{Visited}}

\newcommand{\trace}{\tau}


\newcommand{\shiftLeft}[2]{\mathtt{shl}({#1}:{#2})}
\newcommand{\shiftRight}[2]{\mathtt{shr}({#1}:{#2})}  
\newcommand{\replaceIfEqual}[3]{\left({#1}={#3}?{#2}:{#3} \right)}
\newcommand{\ipost}[3]{\mathtt{post}_{#3}(#1,#2)}
\newcommand{\post}[2]{\mathtt{post}(#1,#2)}

\newcommand{\ipre}[3]{\mathtt{pre}_{#3}(#1,#2)}
\newcommand{\pre}[2]{\mathtt{pre}(#1,#2)}

\newcommand{\bijection}[1]{\theta_{#1}}
\newcommand{\renpc}[2]{\mathtt{ren}(#1,#2)}
\newcommand{\ren}[3]{\mathtt{ren}(#1,#2,#3)}
\newcommand{\renVal}[3]{\mathtt{renameVal}(#1,#2,#3)}
\newcommand{\renPc}[3]{\mathtt{renamePc}(#1,#2,#3)}
\newcommand{\renGraph}[3]{\mathtt{renameGraph}(#1,#2,#3)}
\newcommand{\identity}[1]{\mathtt{id}_{#1}}

\newcommand{\restrictVal}[2]{{#1}_{|#2}}
\newcommand{\restrictPc}[2]{{#1}_{|#2}}
\newcommand{\restrictGraph}[2]{{#1}_{|#2}}

\newcommand{\production}{\mathtt{produce}}
\newcommand{\inspection}{\mathtt{consume}}

\newcommand{\helping}{\mathtt{producing}}
\newcommand{\inspecting}{\mathtt{checking}}

\newcommand{\helper}{\mathtt{producer1}}
\newcommand{\ohelper}{\mathtt{producer2}}
\newcommand{\inspecter}{\mathtt{consumer}}

\newcommand{\waitMode}{\mathrm{\textsc{Wait}}}
\newcommand{\sigMode}{\mathrm{\textsc{Sig}}}
\newcommand{\sigWaitMode}{\mathrm{\textsc{Sig\_Wait}}}



\SetAlFnt{\footnotesize}
\SetProcFnt{\footnotesize}
\SetProcNameFnt{\footnotesize}

\usetikzlibrary{%
  arrows,%
  shapes,%
  positioning,%
  petri,%
  automata,%
  backgrounds%
}

\makeatletter
\newcommand*{\boxwedge}{%
  \mathbin{%
    \mathpalette\@boxwedge{}%
  }%
}
\newcommand*{\@boxwedge}[2]{%
  \sbox0{$#1\boxplus\m@th$}%
  \dimen2=.5\dimexpr\wd0-\ht0-\dp0\relax 
  \dimen@=\dimexpr\ht0+\dp0\relax
  \def\lw{.06}
  \kern\dimen2 
  \tikz[
    line width=\lw\dimen@,
    line join=round,
    x=\dimen@,
    y=\dimen@,
  ]
  \draw
    (\lw/2,0) rectangle (1-\lw,1-\lw)
    (\lw,0) -- (.5,1-\lw-\lw/2) -- (1-\lw-\lw/2 ,0)
  ;%
  \kern\dimen2 
}
\makeatletter

\makeatletter
\newcommand*{\owedge}{%
  \mathbin{%
    \mathpalette\@owedge{}%
  }%
}
\newcommand*{\@owedge}[2]{%
  \sbox0{$#1\oplus\m@th$}%
  \dimen2=.5\dimexpr\wd0-\ht0-\dp0\relax 
  \dimen@=\dimexpr\ht0+\dp0\relax
  \def\lw{.04}
  \def\radius{.5-\lw/2}%
  \kern\dimen2 
  \tikz[
    line width=\lw\dimen@,
    line join=round,
    x=\dimen@,
    y=\dimen@,
    baseline=\dimexpr-.5\dimen@+\dp0\relax,
  ]
  \draw
    (0,0) circle[radius=\radius]
    (225:\radius) -- (0,.5-\lw) -- (-45:\radius)
  ;%
  \kern\dimen2 
}
\makeatother

\newcommand{\intersect}{\owedge}

\newcommand{\headOf}[1]{\mathtt{hd}(#1)}
\newcommand{\tailOf}[1]{\mathtt{tl}(#1)}

\begin{abstract}
  We address the problem of statically checking control state
  reachability (as in possibility of assertion violations, race
  conditions or runtime errors) and plain reachability (as in
  deadlock-freedom) of {\em phaser programs}.
  Phasers are a modern non-trivial synchronization construct that supports
  dynamic parallelism with runtime registration and
  deregistration of spawned tasks. They allow for collective 
  and point-to-point synchronizations.
  For instance, phasers  can 
  enforce barriers or producer-consumer synchronization schemes 
  among all or subsets of the running tasks.
  Implementations 
  are found in modern languages such as X10 or
  Habanero Java.
  Phasers essentially associate phases to individual tasks and use
  their runtime values to restrict possible concurrent executions.
  Unbounded phases may result in infinite transition systems even in
  the case of programs only creating finite numbers of tasks and
  phasers.
  We introduce an exact gap-order based procedure that always terminates
  when checking control reachability for programs generating bounded
  numbers of coexisting tasks and phasers.
  We also show verifying plain reachability is undecidable even
  for programs generating few tasks and phasers.
  We then explain how to turn our procedure into a sound analysis for
  checking plain reachability (including deadlock freedom). 
  We report on preliminary experiments with our open source
  tool. 
\end{abstract}


%

\section{Introduction}
\label{sec:intro}

We focus on safety verification 
of programs using 
{\em phasers} for synchronization \cite{X10:2005,JDVW:phasers:2008,cave2011habanero}.
%
%
This sophisticated 
construct dynamically unifies
collective and point-to-point synchronizations.
%
%
%
For instance, it allows for dynamic registration and deregistration of
tasks allowing for a more balanced usage of the computing resources when
compared to static barriers or producer-consumer constructs
\cite{accumulator:applications:pdp09}.
%
%
%
%
%
The construct can be added to any parallel programming language with a shared 
address space. For instance, it can be found 
in Habanero Java \cite{cave2011habanero}, an extension of the Java
programming language.
%
%
Phasers build on the clock construct from the X10 programming language
\cite{X10:2005}. They can be created dynamically and spawned tasks may
get registered or deregistred at runtime.
%
%

Intuitively, each phaser associates two phases (hereafter {\em wait}
 and {\em signal} phases) to each registered task.
Apart from creating phasers and registering each other to them, tasks can
individually issue $\mathtt{wait}$ and $\mathtt{signal}$ commands to
a phaser they are registered to.
Intuitively, $\mathtt{signal}$ commands are used to inform other
registered tasks the issuing task is done with its {\em
  signal} phase. The command is non-blocking. It increments the {\em
  signal} phase associated to the issuing task on the given phaser.
The $\mathtt{wait}$ command is instead used to check whether all
registered tasks are
done with (i.e., have a {\em signal} phase that is strictly larger than)
the issuing task's {\em wait} phase.
This command may get blocked by a task that did not yet finish the
corresponding phase.
Unlike classical barriers, phasers need not force registered tasks to
wait for each other at each single phase. Instead they allow them to
proceed with the following phases (by issuing $\mathtt{signal}$
commands), or even to exit the construct by deregistering from it.
%
Such dynamic behavior allows for better load balancing and
performance, but comes at the price of making it easy to introduce
programming mistakes such as assertion violations, race conditions,
runtime errors and, in the important situation where wait and signal
commands are decoupled for maximum flexibility, deadlocks.
%
%
%
%
We summarize our contributions in this work: 
\begin{itemize}
\item We propose an operational model based on 
  \cite{JDVW:phasers:2008,cave2011habanero,Cogumbreiro:dynamicverifphasers:2015}.
\item We show 
  undecidability of checking deadlock-freedom 
 for 
  programs with 
  fixed numbers of tasks and phasers.
\item We describe an exact gap-order based symbolic verification
  procedure for checking control state reachability
  (as in assertion violations, race
  conditions or runtime errors)
  and plain reachability
  (as in checking deadlock freedom).
\item We show termination of the procedure for 
  control state
  reachability when numbers of tasks and phasers
  are fixed. 
\item We describe how to turn the procedure into a sound
  over-approximation for plain reachability.
\item We report on our preliminary
  experiments with our open source tool.
\end{itemize}

Related work.
%
We are not aware of automatic formal verification works
that focus on constructs allowing for such a degree of dynamic parallelism.
%
Unlike
\cite{Le:barriers:permissions:2013}, 
we focus on fully automatic verification and consider the richer
and more challenging phaser construct.
%
%
The work of 
\cite{Cogumbreiro:dynamicverifphasers:2015} considers
the dynamic verification of phaser programs and can therefore
only reason about particular program inputs and runs.
The work in \cite{Anderson:jpf:hj:2014} uses Java Path Finder \cite{havelund:jpf:2000}
to explore several runs, but still for one concrete input at a time.
%
%
The works in \cite{FI16:Mayr:Totzke,TCS14:BozzelliPinchinat} target
gap-order systems. Although phaser programs share some of their
properties (larger gaps can do more), the results in \cite{FI16:Mayr:Totzke,TCS14:BozzelliPinchinat}
do not apply since gap-order
systems crucially forbid exact increments.
%
%

Outline.
We describe a phaser program and 
%
recall some preliminaries in Sections~\ref{sec:example} and \ref{sec:prels}. 
This is followed in Section~\ref{sec:language} by a formal description
of phaser programs and of the properties we want to check. 
We also establish 
the undecidability of checking deadlock freedom.
In Section~\ref{sec:symbolic}, we introduce our gap-order based
symbolic representation. In Section~\ref{sec:proc} we describe our  verification
procedure and show decidability of checking 
control state reachability.
%
We then introduce our relaxation procedure for checking plain reachability and briefly describe in
Section~\ref{sec:param} our
applicatoin of view abstraction to the parameterized case.
Finally, we report on our experiments and conclude the work.

\lstset{ %
language=Java,                
basicstyle=\ttfamily\scriptsize,       
numbers=left,                   
numberstyle=\scriptsize,      
stepnumber=1,                   
numbersep=-7pt,                  
showspaces=false,               
showstringspaces=false,         
showtabs=false,                 
frame=single,	                
tabsize=2,	                
mathescape=true,
captionpos=b,                   
breaklines=true,                
breakatwhitespace=false,        
escapeinside={\%*}{*)},          
keywords={bool},
keywordstyle=\color{blue}\ttfamily,
commentstyle=\color{olive}\ttfamily,
emph={newPhaser,signal,wait,asynch,drop},
emphstyle=\color{magenta}\ttfamily
}

\begin{figure*}
\begin{tabular}{c}
\begin{minipage}{.9\textwidth}
\begin{lstlisting}
   bool a, b, done;  
   main()
   {
     done = $\false$;  
     prod = newPhaser($\sigWaitMode$);
     cons = newPhaser($\sigWaitMode$);
     cons.signal();
   
     asynch(aProducer, prod($\sigMode$), cons($\waitMode$));
     asynch(bProducer, prod($\sigMode$), cons($\waitMode$));   
     asynch(abConsumer, prod($\waitMode$), cons($\sigMode$));
   
     prod.drop();
     cons.drop();
   }

   aProducer(p($\sigMode$), c($\waitMode$))
   {
     c.wait();
     while($\neg$done){
       a = $\true$;
       p.signal();
       c.wait();
     };
     p.drop();
     c.drop();
   }
\end{lstlisting}
\end{minipage}
\\
\begin{minipage}{.9\textwidth}
{
  \begin{lstlisting}[firstnumber=28]
    bProducer(p($\sigMode$), c($\waitMode$))
    {
      c.wait();
      while($\neg$done){
        b = $\true$;
        p.signal();
        c.wait();
      };
      p.drop();
      c.drop();
    }
   
   abConsumer(p($\waitMode$), c($\sigMode$))
   {
      while($\neg$done){
       p.wait();
       assert(a $\wedge$ b);
       a = $\false$;
       b = $\false$;

       if($\ndet$)
        done = $\true$;
       c.signal();
     };
     c.drop();
     p.drop();
   }
  \end{lstlisting}
  }
\end{minipage}
\end{tabular}
\caption{Two producers and one consumer are synchronized using two
phasers. In this construction, the consumer requires both producers to
be ahead of it (wrt. \texttt{prod} phaser) in order for it to
consume their respective products. At the same time,
the consumer needs to be ahead of both producers (wrt. \texttt{cons}
phaser) in order for these to be able to produce their pair of
products.}
\label{fig:aggregators}
\end{figure*}

\section{Motivating example}
\label{sec:example}

The program listed in Fig.~(\ref{fig:aggregators}) uses Boolean shared
variables $\bVarSet=\set{\mathtt{a},\mathtt{b},\mathtt{done}}$.
A {\em main} task creates two phasers (lines 5 and 6).  When creating
a phaser, the task gets automatically registered to it. The main task
also creates three other task instances (lines 9, 10 and 11).
Several tasks can be registered to several phasers. 
When a task $\thid$ is registered to a phaser $\phid$, a pair of
numbers $(\waitVal{\thid}{\phid},\sigVal{\thid}{\phid})$, each in
$\nats\cup\set{+\infty}$, is associated to the couple $(\thid,\phid)$.
%
The pair represents the individual
{\em wait} and {\em signal} phases of task $\thid$ on phaser $\phid$.

Registration of a task to a phaser can occur in one of three modes:
$\sigWaitMode$, $\waitMode$ and $\sigMode$.
In $\sigWaitMode$ mode, 
a task may issue
both $\mathtt{signal}$ and $\mathtt{wait}$ commands. 
In $\waitMode$ mode, 
a task may only issue
$\mathtt{wait}$ commands on the phaser. 
%
Finally, when registered in $\sigMode$ mode, 
a task may only issue
$\mathtt{signal}$ commands. 
Issuing a $\mathtt{signal}$ command by a task
on a phaser results in the task
incrementing its signal phase associated to the phaser. This command
is non-blocking.
On the other-hand, issuing a $\mathtt{wait}$ command by a task
on a phaser
$\phid$ will block until {\bf all} tasks registered on $\phid$ 
exhibit signal values on $\phid$ that are strictly
larger than the wait value of the issuing task on phaser $\phid$.
In this case, the wait phase 
of the issuing task is incremented. 
Intuitively, a signal command allows the issuing task to state 
other tasks need not wait for it to complete its 
signal phase. 
In retrospect, a $\mathtt{wait}$ command allows a task to make sure all
registered tasks have moved past its wait phase. 

\begin{figure}[h]
{\scriptsize
\begin{center}
\begin{tikzpicture}
\draw (0,2) node[ellipse, draw=green!50,fill=green!20,thick] (p1)   {$\mathtt{aProducer:@23}$};

\draw (0,-2.3) node[ellipse, draw=green!50,fill=green!20,thick, minimum size = 5mm] (p2)   {$\mathtt{bProducer:@34}$};

\draw (0,1) node[ellipse, draw=green!50,fill=green!20,thick] (m)   {$\mathtt{main:@15}$};

\draw (0,0) node[ellipse, draw=green!50,fill=green!20,thick, minimum size = 5mm] (c)   {$\mathtt{abConsumer:@42}$};

\draw (-4,0) node[rectangle, draw=red!50,fill=red!20,thick] (prodPh)   {$\mathtt{(phaser)}$};

\draw (4,0) node[rectangle, draw=red!50,fill=red!20,thick, minimum size = 1mm] (inspectPh)   {$\mathtt{(phaser)}$};

\path[->]
(m) edge  node[rectangle,draw=white,fill=white] {$\mathtt{prod}$} (prodPh)
(m) edge  node[rectangle,draw=white,fill=white] {$\mathtt{cons}$} (inspectPh)
(p1) edge [bend right] node[rectangle,draw=white,fill=white] {$\mathtt{p}$} (prodPh)
(prodPh) edge  node[rectangle,draw=white,fill=white] {$(\sigVal{}{}=90)$} (p1)
(p2) edge [bend left] node[rectangle,draw=white,fill=white] {$\mathtt{p}$} (prodPh)
(prodPh) edge  node[rectangle,draw=white,fill=white] {$(\sigVal{}{}=91)$} (p2)
(c) edge  [bend left] node[rectangle,draw=white,fill=white] {$\mathtt{p}$} (prodPh)
(prodPh) edge  node[rectangle,draw=white,fill=white] {$(\waitVal{}{}=90)$} (c)
(p1) edge [bend left] node[rectangle,draw=white,fill=white] {$\mathtt{c}$} (inspectPh)
(inspectPh) edge node[rectangle,draw=white,fill=white] {$(\waitVal{}{}=90)$} (p1)
(p2) edge [bend right] node[rectangle,draw=white,fill=white] {$\mathtt{c}$} (inspectPh)
(inspectPh) edge  node[rectangle,draw=white,fill=white] {$(\waitVal{}{}=91)$} (p2)
(c) edge [bend right] node[rectangle,draw=white,fill=white] {$\mathtt{c}$} (inspectPh)
(inspectPh) edge node[rectangle,draw=white,fill=white] {$(\sigVal{}{}=91)$} (c) ;
\end{tikzpicture}
\end{center}
}
\label{fig:cConf}
\caption{Possible wait and signal
phase values for Fig.~(\ref{fig:aggregators}).
Observe that there is no a priori bound on the values of the different wait and signal phases.
In this example, the difference between signal and wait phases is bounded. This is not
always the case in general.
%
}
\end{figure}
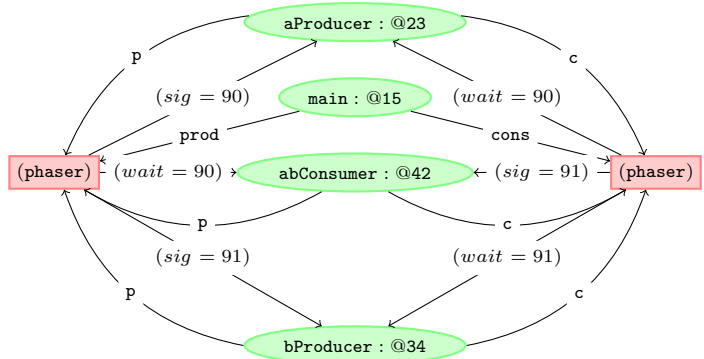

Upon creation of a phaser, wait and signal phases are initialized to
$0$ (except in $\waitMode$ mode where the signal phase is instead initialized
to $+\infty$ in order to not block other waiters).
The only other way a task may get registered to a phaser is if 
an already registred task does register it in the same mode 
(or in $\waitMode$ or $\sigMode$ if the registrar is registered in
$\sigWaitMode$).
In this case, wait and signal phases of the newly registered task are
initialized to those of the registrar. 
Tasks are therefore dynamically registered (e.g., lines
9-11).
They can also dynamically deregister themselves (e.g.,
lines 25-26);

In this example, two producers and one consumer are synchronized using
two phasers. The consumer requires the two producers to be ahead of it
(wrt. the phaser main pointed to with $\mathtt{prod}$) in order
for it to consume their respective products. 
%
At the same time,
the consumer needs to be ahead of both producers (wrt. the 
phaser main pointed to with $\mathtt{cons}$) in order for these to 
produce their pair of
products. 
%
%
It should be clear that phasers can be used as barriers for
synchronizing dynamic subsets of concurrent tasks.
%
%
Observe producers need not, in general, proceed in a lock step
fashion. Producers may produce many items before 
consumers ``catch up''.

%

We are interested in checking: (a) control state reachability as in
assertions (e.g., line 44), race conditions (e.g., mutual exclusion of
lines 20 and 49) or runtime errors (e.g., signaling a dropped
phaser), and (b) plain reachability as in deadlocks (e.g., a producer
at line 23 and a consumer at line 50 with equal phases).
Intuitively, both problems concern themselves with the reachability of
target sets of program configurations. The difference is that control
state reachability defines the targets with the states of the tasks
(their control locations and whether they are registered to some
phasers).
Plain reachability can, in addition,  use values or relations between values
of involved phases. 
Observe that control state reachability depends on the values of the actual
phases, but these values are not used to define the target sets.
For example, assertions are expressed as predicates over boolean variables (e.g., line 44).
Establishing such an assertion requires capturing the constraints imposed by the
phasers on the program behaviors.

Our work proposes a sound and complete algorithm for checking control
state reachability in case a bounded number of tasks and phasers are
generated. The algorithm can handle arbitrarily large phases,
e.g., generated using nested signaling loops.
%
The algorithm  starts from a symbolic representation of all bad
configurations and successively computes sets of predecessor
configurations.
We show termination based on a well-quasi-ordering argument that
imposes restrictions on what can be expressed with our symbolic
representation. For instance  putting upper bounds on differences
between phases is forbidden.
Deadlock configurations cannot be faithfully captured with such restricted
representations. 
Intuitively, a deadlocked configuration will have a cycle where each
involved task is waiting for the task to its right but where the wait
phase of each task equals the signal phase of the task it is waiting
for. 
%
%
We show the problem of checking deadlock freedom to be undecidable
even for programs only generating a bounded number of tasks and phasers. We
explain how to turn our verification algorithm into a sound but
incomplete procedure for checking deadlock-freedom.
Precision can then be augmented on demand to eliminate 
false positives.
%

\section{Preliminaries}
\label{sec:prels}

We use $\nats$ and $\ints$ for natural and
integer numbers respectively.
We write $A\uplus B$ to mean the union of disjoint sets $A$ and $B$.
We let $\pFunctionsOf{A}{B}$ be the set of partial functions from
$A$ to $B$ and use $\emptyFunction{A}{}$ for the empty function over $A$, 
i.e., $\emptyFunction{A}{}(a)$ is undefined (written
$\emptyFunction{A}{}(a)\undefined$) for all $a\in A$. Given function
$g\in\pFunctionsOf{A}{B}$
we write $g(a)\defined$ to mean that
$g(a)$ is defined and write $g\setminus\set{a}$ to mean the restriction of
$g$ to the domain $A\setminus\set{a}$.
%
We write $\mapSubst{g}{a}{b}$ for the function that
coincides with $g$ on $A$ except for $a$ that is sent to $b$.
We abuse  notation and let
$\substSet{g}{\setcomp{a_i\leftarrow b_i}{i\in I}}$,
for a set $\setcomp{a_i}{i\in I}$ of pairwise
different elements, mean the function that coincides with $g$ on
$A$ except for each $a_i$ that is sent to the corresponding $b_i$. 
We sometimes write a function $g$ as a set $\setcomp{a\mapsto
  g(a)}{a\in A}$. It is then implicitly undefined outside of $A$.

\section{Language}
\label{sec:language}


%
A program may use a set $\bVarSet$ of shared boolean variables
and a set $\phVarSet$ of local phaser variables:
%
%
%
%
%
%

{
$$
\begin{array}{lcl}
  \program & ::= & \mathtt{bool}~~ \bvar_1, \ldots, \bvar_{|\bVarSet|};  \\
  & & \task_1(\phvar_{1_1},\ldots,\phvar_{k_1}) \set{\stmt_1} \\
  & & \ldots \\
  & & \task_n(\phvar_{1_n},\ldots,\phvar_{k_n}) \set{\stmt_n}  \\
  \\
  \stmt & ::= & ~~ \phvar = \newp 
  \grammarOr \asynch{\task}{\phvar_1,\ldots,\phvar_k}\\ 
  && 
  \grammarOr \dereg{\phvar} \grammarOr \sig{\phvar}
  \grammarOr \wait{\phvar}  \grammarOr \exit\\
  && 
  ~\grammarOr \stmt;\stmt \grammarOr \bvar=\cond \grammarOr\assert{\cond} \\
  && \grammarOr \while{\cond}{\stmt}
  \grammarOr \ifp{\cond}{\stmt} \\
  \cond & ::= &  \ndet \grammarOr  \true
  \grammarOr \false 
  \grammarOr \bvar
  \grammarOr \cond \vee \cond\\
  && \grammarOr \cond \wedge \cond
  \grammarOr \neg\cond 
\end{array}
$$
}
A program consists in a set of tasks $\tasks$.
A task is declared with $\task(\phvar_{1},\ldots,\phvar_{k})
\set{\stmt}$ where $\phvar_1, \ldots \phvar_k$ are phaser variables
that are local to the declared task. A task can also
create a new phaser with $\phvar=\newp$ and store
the identifier of the phaser in a local variable $\phvar$. We let
$\phVarSet$ be the union of all local phaser variables. 
When creating a phaser, a task gets registered to it.
To simplify our description, we will assume all registrations to be in
$\sigWaitMode$ mode. Including the other modes is a matter of changing
the initial phase values at registration and of statically ensuring the issued
commands respect the registration mode.
%
%
%
%
A task can deregister itself from a phaser referenced by a variable
$\phvar$ with $\dereg{\phvar}$.
It can also issue signal or wait commands on a phaser on which it is registered
and that is referenced by $\phvar$.
A task can spawn another task with
$\asynch{\task}{\phvar_1,\ldots,\phvar_n}$. The issuing task registers
the spawned task to the phasers it points to with
$\phvar_1,\ldots,\phvar_n$.  The issuing task need not wait for the
spawned task and may directly continue  its execution.
%
%

Assume a phaser program $\program=\programTuple$.
We inductively define the finite set $\seqSet$ of control sequences as
follows.
%
%
$\seqSet$ is the smallest set containing:
%
%
(i) suffixes of each ``$\stmt_i$'' appearing in some
``$\task_i(\phvar_{1_i},\ldots,\phvar_{k_i})\set{\stmt_i}$''; and
(ii) suffixes of ``$\stmt_i;\while{\cond}{\stmt_i};\stmt_j$''
(respectively ``$\stmt_i;\while{\cond}{\stmt_i}$'') for each
``$\while{\cond}{\stmt_i};\stmt_i$'' (respectively
``$\while{\cond}{\stmt_i}$'') in $\seqSet$;
and (iii) suffixes of ``$\stmt_i;\stmt_j$'' (respectively ``$\stmt_i$'') for
each ``$\ifp{\cond}{\stmt_i};\stmt_j$''
(respectively ``$\ifp{\cond}{\stmt_i}$'') appearing in $\seqSet$.
%
%
We write $\seqVal$ to mean some control sequence in $\seqSet$, and
$\headOf{\seqVal}$ and $\tailOf{\seqVal}$ to respectively mean the head
and the tail of the sequence $\seqVal$.

\subsection{Semantics.}

A configuration $\cConf$ of 
$\program=\tuple{\bVarSet,\phVarSet,\tasks}$ is a tuple
$\cConfTuple$ where:


\begin{itemize}
\item $\thidSet$ is the current finite set  of
  task identifiers. We let $\thid,\othid$  range over the values in $\thidSet$.
\item $\phidSet$ is the current  finite  set of phaser identifiers. We let
  $\phid,\ophid$  range over the values in
  $\phidSet$.
\item $\boolMp:\bVarSet \to \set{\true,\false}$ is a total mapping
  that associates a value to each $\bvar\in\bVarSet$.
\item $\pcMp:\thidSet \to \seqSet$ is a total mapping that
  associates tasks to their remaining sequences (i.e., control location). 
\item $\varMp:\thidSet \to
  \pFunctionsOf{\phVarSet}{\phidSet}$ is a total mapping that
  associates, to each task identifier in $\thidSet$, a partial
  mapping from the local phaser variables $\phVarSet$ to phaser identifiers
  $\phidSet$. It captures the values of the phaser variables $\phVarSet$
  of each task.
\item $\cPhaseMp:\phidSet\to
  \pFunctionsOf{\thidSet}{\nats^2}$ is a total mapping
  that associates to each phaser $\phid\in\phidSet$ a partial
  mapping $\cPhaseMp(\phid)$ that is defined exactly on the identifiers
  of the tasks registered to $\phid$. For such a task
  $\thid$, $\cPhaseMp(\phid)(\thid)$ is the pair
  $(\waitVal{\thid}{\phid},\sigVal{\thid}{\phid})$ 
  representing wait and signal values of $\thid$ on $\phid$.
\end{itemize}
\begin{figure}[h]
  \begin{flushright}
    \scriptsize
    $$
    \begin{array}{c}
      \infer[\mathtt{(newPhaser)}]
            {\cConfTuple \reducesTo{\thid}{}\tuple{\thidSet,\phidSet',\boolMp,\pcMp',\varMp',\cPhaseMp'}}
            {
              \begin{array}{c}
                \headOf{\pcMp(\thid)}=\phvar:=\newp  ~~\wedge~ \phid\not\in\phidSet ~\wedge~
                \phidSet'=\phidSet\cup\set{\phid} ~\wedge~ \varMp'=\mapSubst{\varMp}{\thid}{\mapSubst{\varMp(\thid)}{\phvar}{\phid}} \\
                \wedge~ \cPhaseMp'=\mapSubst{\cPhaseMp}{\phid}{\set{\thid \mapsto (0,0)}} \wedge \pcMp'=\mapSubst{\pcMp}{\thid}{\tailOf{\pcMp(\thid)}} 
              \end{array}
            }            
            \\
            \\
       \infer[\mathtt{(signal)}]
             {\cConfTuple \reducesTo{\thid}{}
               \tuple{\thidSet,\phidSet,\boolMp,\mapSubst{\pcMp}{\thid}{\tailOf{\pcMp(\thid)}},\cPhaseMp'} }
             {\begin{array}{c}
                 \headOf{\pcMp(\thid)}=\sig{\phvar} ~ \wedge 
                 \varMp(\thid)(\phvar)=\phid ~\wedge~ 
                 \cPhaseMp(\phid)(\thid)=(\waitVal{\thid}{\phid},\sigVal{\thid}{\phid}) ~\wedge \\
                 \cPhaseMp'=\mapSubst{\cPhaseMp}{\phid}
                                     {\mapSubst{\cPhaseMp(\phid)}{\thid}{(\waitVal{\thid}{\phid},1+\sigVal{\thid}{\phid})}}
             \end{array}}
      \\
      \\
      \infer[\left(\!\!\!\!\begin{array}{c}\mathtt{assert.}  \mathtt{ok}\end{array}\!\!\!\!\right)]
            {\cConfTuple \reducesTo{\thid}{}
              \tuple{\thidSet,\phidSet,\boolMp,\mapSubst{\pcMp}{\thid}{\tailOf{\pcMp(\thid)}},\varMp,\cPhaseMp}}
            {
              \begin{array}{c}
                \headOf{\pcMp(\thid)}=\assert{\cond} ~ \wedge ~ 
                \boolMp(\cond)=\true
                \end{array}
            }
            \\
            \\
       \infer[\mathtt{(drop)}]
             {\cConfTuple \reducesTo{\thid}{}
               \tuple{\thidSet,\phidSet,\boolMp,\mapSubst{\pcMp}{\thid}{\tailOf{\pcMp(\thid)}},\varMp,\cPhaseMp'} }
             {\begin{array}{c}
                 \headOf{\pcMp(\thid)}=\dereg{\phvar} ~ \wedge 
                 \varMp(\thid)(\phvar)=\phid ~\wedge~ 
                 \cPhaseMp(\phid)(\thid)\defined ~\wedge~ 
                 \cPhaseMp'=\mapSubst{\cPhaseMp}{\phid}
                                     {\mapSubst{\cPhaseMp(\phid)}{\thid}{\undefined}}
               \end{array}
             }
       \\
       \\
         \infer[\mathtt{(asynch)}]
               {\cConfTuple \reducesTo{\thid}{}
                 \tuple{\thidSet,\phidSet,\boolMp,\mapSubst{\pcMp'}{\thid}{\tailOf{\pcMp(\thid)}},\varMp,\cPhaseMp'}}
               {
                 \begin{array}{c}
                   \headOf{\pcMp(\thid)}=\asynch{\task}{\phvar_1,\ldots \phvar_k}\{\seqVal_1\} ~\wedge~ 
                   \parametersOf{\task}=(\ophvar_1,\ldots \ophvar_k) ~\wedge~ \\
                   \textrm{for each } i:1\leq i\leq k.~\varMp(\thid)(\phvar_i)=\phid_i ~\wedge~ \cPhaseMp(\phid_i)\defined ~\wedge \\  
                   \othid\not\in\thidSet ~\wedge~ 
                   \varMp'=\mapSubst{\varMp}{\othid}{\setcomp{\ophvar_i\mapsto \varMp(\thid)(\phvar_i)}{1\leq i\leq k}} ~\wedge~ 
                   \pcMp'=\mapSubst{\pcMp}{\othid}{\seqVal_1} ~\wedge~ \\
                   \cPhaseMp'=\substSet{\cPhaseMp}
                                       {\setcomp{\phid_i\leftarrow
                                           \mapSubst{\cPhaseMp(\phid_i)}{\othid}{\cPhaseMp(\phid_i)(\thid)}
                                         }
                                         {\varMp(\thid)(\phvar_i)=\phid_i \textrm{ for } \phid_i\in\phidSet\textrm{ and }1\leq i\leq k}}
                 \end{array}
               }
       \\
       \\
                     \infer[\mathtt{(wait)}]
                           {\cConfTuple \reducesTo{\thid}{}
                             \tuple{\thidSet,\phidSet,\varMp,\mapSubst{\pcMp}{\thid}{\tailOf{\pcMp(\thid)}},\cPhaseMp'} }
                           {\begin{array}{c}
                               \headOf{\pcMp(\thid)}=\wait{\phvar} ~\wedge~
                               \varMp(\thid)(\phvar)=\phid ~\wedge~\cPhaseMp(\phid)(\thid)=(\waitVal{\thid}{\phid},\sigVal{\thid}{\phid}) ~\wedge~
                               \\
                               \forall \othid\in\thidSet.~
                               \left(\cPhaseMp(\phid)(\othid)=(\waitVal{\othid}{\phid},\sigVal{\othid}{\phid})
                               \Rightarrow \waitVal{\thid}{\phid} < \sigVal{\othid}{\phid} \right) ~\wedge~
                               \cPhaseMp'=\mapSubst{\cPhaseMp}{\phid}
                                                   {\mapSubst{\cPhaseMp(\phid)}{\thid}{(1+\waitVal{\thid}{\phid},\sigVal{\thid}{\phid})}}
                           \end{array}}
       \\
       \\
                     \infer[\mathtt{(exit)}]
                           {\tuple{\thidSet,\phidSet,\boolMp,\pcMp,\varMp,\cPhaseMp} \reducesTo{\thid}{}
                             \tuple{\thidSet\setminus\set{\thid},\phidSet,\boolMp,\pcMp',\varMp',\cPhaseMp'}}
                           {
                             \begin{array}{c}
                               \headOf{\pcMp(\thid)}=\exit ~\wedge~ 
                               \varMp'=\varMp\setminus\set{\thid} ~\wedge~ 
                               \pcMp'=\pcMp\setminus\set{\thid} ~\wedge~ 
                               \cPhaseMp'=\substSet{\cPhaseMp}{\setcomp{\phid \leftarrow (\cPhaseMp(\phid)\setminus\set{\thid})}{~\phid\in\phidSet}}
                             \end{array}
                           }
       \\
       \\
%
       \infer[\left(\!\!\!\!\begin{array}{c}\mathtt{assert.} \\ \mathtt{fault}\end{array}\!\!\!\!\right)]
            {\cConfTuple \in 
              \cConfSetAssert{|\thidSet|}{|\phidSet|}}
            {
              \begin{array}{c}
              \headOf{\pcMp(\thid)}=\assert{\cond} ~ \wedge ~
              \boolMp(\cond)=\false
              \end{array}
            }
            \\
            \\
         \infer[\left(\!\!\!\!\begin{array}{c}\mathtt{runtime} \\ \mathtt{error} \end{array}\!\!\!\!\right)]
               {\cConfTuple \in
                 \cConfSetRuntime{|\thidSet|}{|\phidSet|} }
               {\begin{array}{c}
                     \headOf{\pcMp(\thid)}=\seqVal ~\wedge 
                   \left(
                   \seqVal=\dereg{\phvar} ~  
                   \vee ~  \seqVal= \sig{\phvar} ~ \right. \\
                   \left. \vee ~\seqVal =\wait{\phvar} ~
                   \vee ~\seqVal = \asynch{\task}{\ldots,\phvar,\ldots}~
                   \right) 
                   ~\wedge~ 
                   \left( \varMp(\thid)(\phvar)\undefined ~\vee~  
                   \cPhaseMp(\varMp(\thid)(\phvar))(\thid)\undefined \right)
               \end{array}}
                      \\
                      \\                      
       \infer[(\mathtt{deadlock})]
               {\cConfTuple \in
                 \cConfSetDeadlock{|\thidSet|}{|\phidSet|} }
               {\begin{array}{c}
                   \set{\thid_0,\ldots\thid_n}\subseteq\thidSet \wedge \set{\phid_0,\ldots\phid_n}\subseteq\phidSet ~\wedge ~ 
                   \forall ~ i:0\leq i\leq n. ~ 
                   \headOf{\pcMp(\thid_i)}=\wait{\phvar_i}
                   ~\wedge~ \\
                   \varMp(\thid_i)(\ophvar_i)=\phid_{(i+1)\%n} 
                   \wedge 
                   \varMp(\thid_i)(\phvar_i)=\phid_i ~\wedge~  
                    \waitVal{\thid_i}{\phid_{(i+1)\%n}} \geq \sigVal{\thid_{(i+1)\%n}}{\phid_{(i+1)\%n}}
               \end{array}}
               \\
               \\
       \infer[\mathtt{(race)}]
             {\cConfTuple \in
               \cConfSetRace{|\thidSet|}{|\phidSet|}}
             {
               \begin{array}{c}
               \headOf{\pcMp(\thid)}=\bvar:=\cond ~ \wedge \headOf{\pcMp(\othid)}=\seqVal' \wedge~ \thid \neq \othid~\wedge \\
               \left(
               \begin{array}{c}
                 \seqVal'=\bvar:=\cond' \vee  \bvar \textrm{ appears in } \cond' \textrm{ and }  \\
                 \left(
                 \begin{array}{c}
                   (\seqVal'=\ifp{\cond'}{\stmt} \vee  \\
                   \seqVal'=\while{\cond'}{\stmt'} \vee \\
                   \seqVal'=\assert{\cond'} \vee \seqVal'=\bvar'=\cond')
                 \end{array}
                 \right)
               \end{array}
               \right)
               \end{array}
             }
    \end{array}
    $$
  \end{flushright}

  \caption{Operational semantics of phaser statements.}
  \label{fig:semantics}
\end{figure}

The set of tasks $\thidSet$ is altered by
$\asynch{\task}{\phvar_1,\ldots,\phvar_n}$ and $\exit$ statements
(rules $\mathtt{(asynch)}$ and $\mathtt{(exit)}$ in Fig.(\ref{fig:semantics}) of the appendix).
The set of phasers $\phidSet$ is updated upon creation of
new phasers (rule $\mathtt{(newPhaser)}$ in Fig.(\ref{fig:semantics})) of the appendix. 
The mapping $\varMp$ associates values to program phaser variables. 
Accessing variables with undefined values, or phasers to which
the task is not currently registered, leads to runtime errors (rule $\mathtt{(runtime~error)}$).
The total mapping $\cPhaseMp$ captures states of phasers. 
It associates to each phaser identifier $\phid$ in $\phidSet$ a
partial mapping $\cPhaseMp(\phid)$. This partial mapping is defined
for a task identifier $\thid\in\thidSet$ (i.e.,
$\cPhaseMp(\phid)(\thid)\defined$) iff the task $\thid$ is registered
to the phaser $\phid$. In this case, $\cPhaseMp(\phid)$ gives the
waiting phase $\waitVal{\thid}{\phid}$ and the signaling phase
$\sigVal{\thid}{\phid}$ of the task $\thid$ on the phaser $\phid$.
Initially, a unique ``main'' task $\mathtt{t_0}$ starts executing its
$\stmt_{main}$ with no phasers. $\cPhaseMp$ is the empty function with
an empty domain $\emptyFunction{\emptyset}{}$.  After a task $\thid$
executes a $\phvar := \newp$ statement (rule $\mathtt{(newPhaser)}$ in
Fig.(\ref{fig:semantics}))  of the appendix, a new phaser $\phid$ is associated to the
variable $\phvar$ using $\varMp$ and $\cPhaseMp(\phid)$ becomes the
partial function $\set{\thid\mapsto (0,0)}$.
%
%
The initial configuration is 
$\cConf_{init}=\tuple{\set{\mathtt{t_0}},\set{},\boolMp_{\false},\set{\mathtt{t_0}\mapsto \stmt},\emptyFunction{}{},\emptyFunction{}{}}$,
where a ``main'' task with identifier $\mathtt{t_0}$ and code $\stmt$ is the unique initial
task. No phasers are present in the initial configuration,
and all boolean variables are mapped to $\false$.

Given two configurations $\cConf$ and $\cConf'$ with
$\cConf=\cConfTuple$, we write $\cConf\reducesTo{\thid}{}\cConf'$ if there
is a task $\thid\in\thidSet$ such that one of the rules in
Fig.(\ref{fig:semantics})  of the appendix holds. 
We use $\reducesTo{*}{}$ for the reflexive transitive closure of $\reducesTo{}{}$
and write $\cConf\reducesTo{*}{}\cConf'$ to mean that $\cConf'$ is reachable from
$\cConf$.
A configuration is said reachable if it is reachable from the initial
configuration $\cConfInit$.

\subsubsection{Control-state reachability}
%
%
%
Checking the possibility of assertion violations, 
of runtime errors and of race conditions amounts to checking reachability
of configurations respectively in $\cConfSetAssert{n}{p}$, $\cConfSetRuntime{n}{p}$ and in $\cConfSetRace{n}{p}$
for some number of tasks $n$ and number of phasers $p$.
We introduce in Section \ref{sec:symbolic} a complete procedure
for checking reachability of such sets of configurations and show it to be sound
for programs with a fixed
upper bounds on numbers of generated phasers and tasks.

\subsubsection{Plain reachability and deadlocks.}
We are also interested in checking the possibility of deadlocks.
For this we need to define the notion of a blocked
task.
Assume in the following a configuration $\cConf=\cConfTuple$.  
\begin{definition}[Blocked]
  A task $\thid\in\thidSet$ is blocked at phaser $\phid\in\phidSet$ by
  task $\othid\in\thidSet$ if $\headOf{\pcMp\mathtt{(\thid)}}={\wait{\phvar}}$
  with $\varMp(\thid)({\phvar})=\phid$ and
  $\cPhaseMp(\phid)(\thid)=(\waitVal{\thid}{\phid},\_)$ when
  $\cPhaseMp(\phid)(\othid)=(\_,\sigVal{\othid}{\phid})$ and
  $\sigVal{\othid}{\phid} \leq \waitVal{\thid}{\phid} $.
\end{definition}

Intuitively, a task $\thid$ is blocked by a task $\othid$ if  it
cannot finish its $\mathtt{wait}$ command on some phaser 
because it is waiting for task $\othid$ that did not
issue enough $\mathtt{signal}$ commands on the same phaser.

\begin{definition}[Deadlock]
  $\cConfTuple$ is a deadlock configuration if each task of a non empty 
  subset $\othidSet\subseteq\thidSet$ is blocked by some task in $\othidSet$.
\end{definition}

\begin{theorem}[Deadlock-Freedom]
  \label{th:reachability}
  It is undecidable in general, even for programs with only
  three phasers and four tasks,
  to check for deadlock-freedom. 
\end{theorem}
\begin{figure}[ht]
{
\begin{lstlisting}[basicstyle=\footnotesize\ttfamily,numbers=none,mathescape=false,commentstyle=\color{blue}\ttfamily]
q in {s0, s1, ..., sF}
com in {done, check, t12_dec_x1,t12_inc_x1, t12_reset_x1,t23_dec_x2, t23_inc_x2, ...};

main(){
 v1 = newPhaser();
 v2 = newPhaser();
 v3 = newPhaser();

 q = s0;
 com = done;

 asynch(task12,v1,v2,v3);
 asynch(task23,v2,v3,v1);
 asynch(task31,v3,v1,v2);

 v1.drop();
 v2.drop();
 v3.drop();

 while(true){
  //check reachability
  if(( * ) && (q == sF)){
   com = check;
   exit;
  };

  // let f,g,h be strings automorphisms on strings
  // that coincide with the identity except for
  // that (f($a)=1, f($b)=2 and f($c)=3), (g($a)=2,
  // g($b)=3 and g($c)=1) and (h($a)=3, h($b)=1
  // and h($c)=2). Syntactically apply each
  // morphism to the remaining part of the while loop.
  
  //case (si,inc(x$a),sj) 
  if (( * ) && (q == si)){
   com = t($c$a)_inc_x($a);
   while(com != done){};
   q = sj;
  };  
  //case (si,dec(x$a),sj)
  if(( * ) && (q == si)){
   com = t$a$b_dec_x$a;
   while(com != done){};
   q = sj;
  };    
  //case (si,reset(x$a),sj)
  if(( * ) && (q == si)){
   com = t$a$b_reset_x$a;
   while(com != done){};
   q = sj;
  };
 };
}
\end{lstlisting}
}
\caption{The main task orchestrates the encoding of the reachability
  of any three counters reset machine as reachability of a deadlock
  configuration of a phaser program. }
\label{fig:deadlock1}
\end{figure}
\begin{figure}[h]
  {
\begin{lstlisting}[basicstyle=\footnotesize\ttfamily,numbers=none,mathescape=false,commentstyle=\color{blue}\ttfamily]
  // Syntactically apply each morphism in {f,g,h} to the whole body
  // text of the task. This gives three tasks task12(v1,v2,v3),
  // task23(v2,v3,v1) and task31(v3,v1,v2)

task$a$b(v$a,v$b,v$c){
 while(true){
  if(com = t$a$b_inc_x$b){
   v$b.sig();
   com = done;
  };
  if(com = t$a$b_dec_x$a){
   v$a.sig();
   com = t$b$c_dec_x$a;
   v$a.wait();
   while(com != t$a$b_dec_x$a){};
   com = done;
  };
  if(com = t$a$b_dec_x$b){
   v$b.wait();
   com = t$c$a_dec_x$b;
  };
  if(com = t$a$b_dec_x$c){
   v$c.sig();
   com = t$b$c_dec_x$c;   
   v$c.wait();
   while(com != t$a$b_dec_x$c){}
   com = t$c$a_dec_x$c;
  };
  if(com = t$a$b_reset_x$a){
   v$a.sig();
   com = t$b$c_reset_x$a;
   while(com != t$a$b_reset_x$a){};
   v$a.wait();
   asynch(taskc$a,v$c,v$a,v$b);
   com = done;
  };
  if(com = t$a$b_reset_x$b){
   com = t$b$c_reset_x$b;
   exit; 
  };
  if(com = t$a$b_reset_x$c){
   v$c.sig();
   com = t$b$c_reset_x$c;
   while(com != t$a$b_reset_x$c){};
   v$c.wait();
   com = t$c$a_reset_x$b;
  };
  if(com = check){
   v$a.sig();
   v$a.wait();
  };
 };
}
\end{lstlisting}
}
\caption{Three tasks: task12, task23 and task31 synchronize with the
  main task in the encoding of the reachability of any three counters
  reset machine as reachability of a deadlock configuration of a
  phaser program. }
\label{fig:deadlock2}
\end{figure}
\begin{proof}
  Sketch.  We encode the reachability problem of any given 3-counters
  reset-VAS (vector addition system with reset arcs)
  as the reachability problem of a configuration with a simple cycle involving
  three tasks.
  Indeed,  reachability of configuration
  $(s_F,0,0,0)$ (three counters $x,y,z$ with zero values at some control
  location $s_F$) is undecidable for reset-VASs.
  Figures ~(\ref{fig:deadlock1}-\ref{fig:deadlock2}) in the appendix describe a
  phaser program where a main task $\thid$ spawns three tasks
  $\set{\thid_{12},\thid_{23},\thid_{31}}$ s.t. $\thid_{\$a\$b}$ runs
  $\mathtt{task\$a\$b}$ for each
  $\mathtt{\$a\$b}\in\set{\mathtt{12},\mathtt{23},\mathtt{31}}$.
  The main task orchestrates the simulation and is registered to
  three phasers $\phid_1,\phid_2$ and $\phid_3$. Main and the other
  tasks use their respective local variable $\mathtt{v\$i}$, for
  $i\in\set{1,2,3}$, to point to phaser $\phid_{\$i}$.
  %
  %
  The idea is to encode the value of counter $x_{\$b}$ using the
  difference $\sigVal{\thid_{\$a\$b}}{\phid_{\$b}}-\waitVal{\thid_{\$b\$c}}{\phid_{\$b}}$, for
  $\$a\$b\$c \in \set{123,231,312}$.
  %
  %
  Resets of counter $x_{\$b}$ are encoded by asking
  task $\thid_{\$a\$b}$ to exit (hence deregistering from all phasers)
  and having task $\thid_{\$b\$c}$ spawn a new $\mathtt{task\$a\$b}$. 
  Finally asking each task $\thid_{\$a\$b}$ to perform a
  $\mathtt{wait}$ on phaser ${\phvar_{\$a}}$ ensures a simple cycle
  of size 3 is built exactly when the three counters are $0$.
\end{proof}

\section{Symbolic verification of phaser programs}
\label{sec:symbolic}

We briefly introduce gap-order constraints and use
them to define a symbolic representation (hereafter constraints) that we use in Section
\ref{sec:proc} for checking reachability.

\subsection{Gap-order constraints and
  graphs~\cite{TCS93:Revesz,LahiriMusuvathi:utvpi:2005,TCS14:BozzelliPinchinat,FI16:Mayr:Totzke}.} 
Gap-order constraints can be regarded as a particular case of the octagons or the
\emph{unit two variables per inequality} (utvpi) constraints. 
Assume in this section that $\xvar$ and $\yvar$
are integer variables and that $\kconst$ is an integer constant. We use
$\xSet$ and $\ySet$ to mean finite sets of integer
variables.
A valuation $\val$ is a total function $\xSet\to\ints$. 
Valuations are implicitly extended to preserve
constants (i.e. $\val(\kconst)=\kconst$ for any $\kconst\in\ints$).
A {\em gap-order clause} $\gclause$ over $\xSet$ is an
inequality of the form $a-b\geq k$ where $a,b\in\xSet\cup\set{0}$. A
{\em gap-order constraint} $\gcstr$ over $\xSet$ is a finite
conjunction of gap-order clauses over the same set
$\xSet$.
Observe that $\left(\xvar = \yvar+2 ~\wedge~ \yvar \leq 5\right)$ is essentially a
gap-order constraint because it can be equivalently rewritten as the
conjunction $\left(\xvar - \yvar \geq 2 ~\wedge~ \yvar -\xvar \geq -2 ~\wedge~ 0
- \yvar \geq -5\right)$.
Given a gap-order constraint $\gcstr$ over $\xSet$ and a
valuation $\val:\xSet\to\ints$,  we write
$\val\models\gcstr$ to mean that $\val(\ael)-\val(\bel)\geq k$
holds for each gap-order clause $\gclause: \ael - \bel \geq k$
appearing in $\gcstr$.
We let $\satOf{\gcstr}$ be the set $\setcomp{\val:\xSet\to\ints}{\val\models\gcstr}$.

A gap-order graph (or graph for short) $\ggraph$ over
$\xSet$ is a graph $\ggraphTuple$ with vertices $\vertices=\xSet\cup\set{0}$
where edges in $\edges$ are of the form $\gedge{\ael}{\kconst}{\bel}$
with $\ael,\bel\in\vertices$ and
weight $\kconst$ in $\ints\cup\set{-\infty,+\infty}$.
We let $\verticesOf{\ggraph}=\xSet$.
Given a gap-constraint $\gcstr$ over $\xSet$, 
we can build the graph $\ggraphOf{\gcstr}$ with vertices
$\xSet\cup\set{0}$ and where $\edges$ only contains a representative
$\gedge{\ael}{\kconst}{\bel}$ edge for each clause $\ael - \bel \geq
k$ appearing in $\gcstr$.
A valuation $\val:\xSet\to\ints$ satisfies a graph $\ggraph=\ggraphTuple$
(written $\val\models\ggraph$) iff
$\val(\ael)-\val(\bel)\geq \kconst$ for each $\gedge{\ael}{\kconst}{\bel}\in\edges$.
We let $\satOf{\ggraph}$ be the set $\setcomp{\val:\xSet\to\ints}{\val\models\ggraph}$.
Clearly, $\satOf{\ggraphOf{\gcstr}}=\satOf{\gcstr}$.
The closure $\closureOf{\ggraph}$ of a graph $\ggraph=\ggraphTuple$ is the unique
complete graph with the same vertices $\vertices$ and where
$\gedge{a}{k'}{b}$ is an edge of $\closureOf{\ggraph}$ iff
$k'\in\ints\cup\set{-\infty,+\infty}$ 
is the least upper bound of all weight-sums for any path in $\ggraph$
from $\ael$ to $\bel$.
Closure allows us to deduce $(0-x \geq -7)$ from $(y-x\geq -2 \wedge
0-y\geq -5)$.
The result of the closure procedure is a special graph $\emptyGraph$ denoting
the graph without any satisfying valuation
each time a weight k=$+\infty$ is generated.
The closure of a graph can be computed in
polynomial time and we get $\satOf{\closureOf{\ggraph}}=\satOf{\ggraph}$.
We define the {\em degree} of a graph $\ggraph$ (written $\degreeOf{\ggraph}$)
to be $0$ if no edge in $\closureOf{\ggraph}$ has a
negative weight apart from $-\infty$. Otherwise, $\degreeOf{\ggraph}$ is the largest
natural $k\in\nats$ such that there is an edge in
$\closureOf{\ggraph}$ with weight $-k$.
For instance, the degree of the graph resulting from $(x-y\geq 2 \wedge y-x\geq -4)$ is $4$.
We systematically close all manipulated graphs and write $\ggraphs(\xSet)$
for the set of closed graphs over $\xSet$.
Given a graph $\ggraph$,  we write $\subst{\ggraph}{\xvar}{\yvar}$ to
mean the graph obtained by replacing the vertex $\xvar$ by the vertex
$\yvar$. We abuse notation and write
$\substSet{\ggraph}{\setcomp{\xvar_i/\yvar_i}{i\in I}}$, for pairwise
different $\xvar_i$ elements to mean the simultaneous application of
the individual substitutions.
For a set of variables $\ySet$,  we write
$\ggraph\removeVars\ySet$ to mean the graph obtained by
removing the variables in $\ySet$ from the vertices
of $\ggraph$.
Given two closed graphs $\ggraph$ and $\oggraph$ over the same $\xSet$,
we write $\ggraph \entailed \oggraph$ to mean that each directed edge in
$\ggraph$ is labeled with a larger weight in $\oggraph$. 
As a result, $\satOf{\oggraph}\subseteq\satOf{\ggraph}$.
Finally, we write $\ggraph \intersect \oggraph$ to mean the closure
of the graph obtained with merging the two sets of vertices and edges. 
As a result, $\satOf{\ggraph\intersect\oggraph}=\satOf{\ggraph}\cap\satOf{\oggraph}$.

\subsection{Constraints as a symbolic representation.}
A constraint $\sConf$ is a tuple $\sConfTuple$ where the
only difference with the definition of a configuration $\cConfTuple$
is the adoption of a gap-order constraint 
$\sPhaseMp$ instead of $\cPhaseMp$.
More specifically, $\sPhaseMp:\phidSet\to
\cup_{\othidSet\subseteq\thidSet}\ggraphs(\cup_{\thid\in\othidSet}\smset{\waitVar{\thid}{\phid},\sigVar{\thid}{\phid}})$
is a total mapping that associates a gap-order graph
%
to each phaser $\phid\in\phidSet$.
  Intuitively, we use variables $\waitVar{\thid}{\phid}$ and
  $\sigVar{\thid}{\phid}$ to constrain in graph $\sPhaseMp(\phid)$
  possible values of both wait ($\waitVal{\thid}{\phid}$) and signal
  ($\sigVal{\thid}{\phid}$) phases of each task $\thid$ registered
  to phaser $\phid$.
%
%
  As a result, we can check if task $\thid$ is registered to
  phaser $\phid$ according to graph $\ggraph=\sPhaseMp(\phid)$ by checking if
  $\smset{\waitVar{\thid}{\phid},\sigVar{\thid}{\phid}}\subseteq\verticesOf{\ggraph}$.
  %
  We will write $\registred{\phid}{\ggraph}$ to mean the set of tasks
  $\smsetcomp{\thid}{\smset{\waitVar{\thid}{\phid},\sigVar{\thid}{\phid}}\subseteq\verticesOf{\ggraph}}$.
  We also write $\isRegistred{\thid}{\phid}{\ggraph}$ for
  the predicate $\thid\in\registred{\phid}{\ggraph}$.
  Observe that the language semantics impose that, for each phaser
  $\phid$ and for any pair $\thid,\othid$ of tasks in
  $\registred{\phid}{\ggraph}$, the predicate $0 \leq
  \waitVal{\thid}{\phid} \leq \sigVal{\othid}{\phid}$ is an invariant. 
  For this reason, we always safely strengthen, 
  in any obtained $\sPhaseMp(\phid)=\ggraph$,
  weights $\kconst$ in 
  $\gedge{\sigVar{\thid}{\phid}}{\kconst}{\waitVar{\othid}{\phid}}$,
  $\gedge{\sigVar{\thid}{\phid}}{\kconst}{0}$ and
  $\gedge{\waitVar{\thid}{\phid}}{\kconst}{0}$  
  with $max(\kconst,0)$.
  %
  The following definition helps us characterize
  configurations for which our procedure 
  terminates.

  \begin{definition}[degree and freeness of constraints]
  A constraint
  $\tuple{\thidSet,\phidSet,\boolMp,\pcMp,\varMp,\sPhaseMp}$ has as degree
  the largest degree among all its graphs $\sPhaseMp(\phid)$ for $\phid\in\phidSet$
  if $\phidSet$ is not empty and $0$ otherwise. 
  %
  Furthermore, a constraint is said to be {``free''} if, for
  any $\phid\in\phidSet$, the only edges in $\sPhaseMp(\phid)$ with
  weights different from $-\infty$ are edges of the forms (i)
  $\gedge{\sigVar{\thid}{\phid}}
  {\kconst_{(\sigVar{\thid}{\phid},\waitVar{\othid}{\phid})}}{\waitVar{\othid}{\phid}}$,
  (ii) 
  $\gedge{\sigVar{\thid}{\phid}}{\kconst_{(\sigVar{\thid}{\phid})}}{0}$,
  or (iii) 
  $\gedge{\waitVar{\thid}{\phid}}{\kconst_{(\waitVar{\thid}{\phid})}}{0}$ for some
  $\thid,\othid\in\registred{\phid}{\sPhaseMp(\phid)}$ and
  $\kconst_{(\sigVar{\thid}{\phid},\waitVar{\othid}{\phid})},
  \kconst_{(\sigVar{\thid}{\phid})},
  \kconst_{(\waitVar{\thid}{\phid})}
  \in\nats$
  \end{definition}

Free constraints are only allowed to impose, for the same phaser,
non-negative lower bounds on differences between signals and waits,
between signals and 0, and between waits and 0. Like
degree-0-constraints, free constraints are not allowed to put a
positive upper bound on how much a signal is larger than a wait.
Unlike degree-0-constraints, they are not allowed to put bounds on the
differences among signal values, or among wait values. For instance a
free constraint cannot impose
$\sigVar{\thid}{\phid}-\sigVar{\othid}{\phid} = 0$ while a
degree-0-constraint can. Intuitively, freeness does not oblige our
verification procedure to maintain exact differences when firing
"signal" or "wait" instructions, jeopardizing termination.
This will be stated in Section \ref{sec:proc}.


  \subsection{Denotations of constraints.}
  Given a configuration
  $\cConf=\tuple{\thidSet,\phidSet,\boolMp, \pcMp, \varMp, \cPhaseMp}$
  and a constraint
  $\sConf=\tuple{\thidSet',\phidSet',\boolMp',\pcMp',\varMp',\sPhaseMp'}$,
  we say that $\cConf$ satisfies $\sConf$, and write $\cConf \models
  \sConf$, if $\cConf$ satisfies (up to a renaming of the tasks and the phasers)
  conditions  imposed by $\sConf$.
  More concretely, $\cConf \models
  \sConf$ if $\boolMp=\boolMp'$ and there are bijections 
  $\taskMapping: \thidSet\to\thidSet'$
  and $\phaserMapping: \phidSet\to\phidSet'$ such that:
  (i)   $\pcMp(\thid)=\pcMp'(\taskMapping(\thid))$ for each
  $\thid\in\thidSet$; and (ii)
  $\phaserMapping(\varMp(\thid)(\phvar))=\varMp'(\taskMapping(\thid))(\phvar)$
  for  each $\thid\in\thidSet$ and $\phvar\in\phVarSet$;
  and (iii) the renaming of tasks and
  phasers in $\cPhaseMp$ wrt. $\taskMapping$ and $\phaserMapping$
  satisfies $\sPhaseMp$, i.e., (iii.a)
  for each $\thid\in\thidSet$ and each $\phid\in\phidSet$,
  $\cPhaseMp(\phid)(\thid)\defined$ iff
  $\isRegistred{\taskMapping(\thid)}{\phaserMapping(\phid)}{\sPhaseMp(\phaserMapping(\phid))}$, and (iii.b)
  for each $\phid'\in\phidSet'$,
  $\ggraph(\bigwedge_{\thid'\in\registred{\phid'}{\sPhaseMp(\phid')}}((\waitVar{\thid'}{\phid'},\sigVar{\thid'}{\phid'})=\cPhaseMp(\phaserMapping^{-1}(\phid'))(\taskMapping^{-1}(\thid'))))
  \models\sPhaseMp(\phid')$.
  %
  %
  %
  %
 %
  We let $\denotationOf{\sConf}$ denote  $\setcomp{\cConf}{\cConf\models\sConf}$.
%
  %
  Intuitively, $\denotationOf{\sConfTuple}$ contains all
  configurations $\cConf$ with the same number of tasks and phasers
  and such that there are renamings of tasks and phasers that preserve
  in $\cConf$ the correspondence between $\pcMp$, $\varMp$ and
  $\sPhaseMp$.
  We write $\denotationOf{\sConfSet}$, for a set $\sConfSet$ of constraints,
  to mean the union $\cup_{\sConf\in\sConfSet}\denotationOf{\sConf}$.
  Given a program $\programTuple$, 
  we can exactly characterize with
  a finite set of constraints all configurations involving $n$
  tasks and $p$ phasers and satisfying the premises of rules
  $\mathtt{(runtime~error)}$, $\mathtt{(assert.~fault)}$,
  $\mathtt{(race)}$ and $\mathtt{(deadlock)}$
  from Fig.(\ref{fig:semantics})  of the appendix.

  \begin{lemma}[Characterizing badness]
    Given a program $\programTuple$ and natural numbers $(n,p)$,
    we can exhibit finite sets of constraints
    $\sConfSetRace{n}{p}$, $\sConfSetAssert{n}{p}$,
    $\sConfSetRuntime{n}{p}$ and $\sConfSetDeadlock{n}{p}$
    such that:
    \begin{align*}
    \cConfSetRace{n}{p}&=\denotationOf{\sConfSetRace{n}{p}} \\
    \cConfSetAssert{n}{p}&=\denotationOf{\sConfSetAssert{n}{p}}\\
    \cConfSetRuntime{n}{p}&=\denotationOf{\sConfSetRuntime{n}{p}} \\
    \cConfSetDeadlock{n}{p}&=\denotationOf{\sConfSetDeadlock{n}{p}}
    \end{align*}
    In addition, we can choose the constraints in
    $\sConfSetDeadlock{n}{p}$ to be of degree $0$ while those in
    $\sConfSetRace{n}{p}$, $\sConfSetAssert{n}{p}$ or in
    $\sConfSetRuntime{n}{p}$ to be free. 
  \end{lemma}
  \begin{proof}
    Observe that $n$ and $p$ are given naturals. Fix a task set
    $\thidSet=\set{\thid_1,\ldots,\thid_n}$ of size $n$ and a phaser set
      $\phidSet=\set{\phid_1,\ldots,\phid_p}$ of size $p$.
    We can therefore enumerate 
    all tuples $(\thidSet, \phidSet, \boolMp, \pcMp, \varMp)$
    where:
    \begin{itemize}
      \item $\boolMp:\bVarSet \to \set{\true,\false}$ is a total mapping
  that associates a value to each $\bvar\in\bVarSet$.
\item $\pcMp:\thidSet \to \seqSet$ is a total mapping that
  associates tasks to a sequence $\seqVal\in\seqSet$
\item $\varMp:\thidSet \to
  \pFunctionsOf{\phVarSet}{\phidSet}$ is a total mapping that
  associates, to each task identifier in $\thidSet$, a partial
  mapping from the local phaser variables $\phVarSet$ to phaser identifiers
  $\phidSet$. 
      \end{itemize}
    Let $\mathcal{C}$ be the set of such tuples.
    Observe that tuples in $\mathcal{C}$ are missing
    information about tasks' registration and wait and signal values.
    We complete this informaiton in the following.
    %
    Given a set $\othidSet\subseteq\thidSet$ of tasks and
    a phaser $\phid\in\phidSet$, we write
    $\topOf{\othidSet,\phid}$ to mean the graph of the conjunction
    $\bigwedge_{\thid,\othid\in\othidSet}(\sigVar{\thid}{\phid} \geq \waitVar{\othid}{\phid}\geq 0)$.
    Observe that an invariant of all phaser programs is that signal and
    wait phases (of all tasks registered on a given phaser) are always
    non-negative with the formers always larger or equal than the
    laters.
    For this reason,  $\topOf{\othidSet,\phid}$
    is the weakest possible graph
    where the set $\othidSet$ is registered to a phaser $\phid$.
    Observe that $\topOf{\othidSet,\phid}$ is free.
    We now finish the definitions of $\sConfSetRace{n}{p}$,
    $\sConfSetAssert{n}{p}$, $\sConfSetRuntime{n}{p}$,
    $\sConfSetDeadlock{n}{p}$:
    \begin{itemize}
    \item Add to $\sConfSetRace{n}{p}$ all constraints
      $\tuple{\thidSet,\phidSet,\boolMp,\pcMp,\varMp,\sPhaseMp}$
      where:
      \begin{itemize}
      \item $\tuple{\thidSet,\phidSet,\boolMp,\pcMp,\varMp}$ is a tuple of 
      $\mathcal{C}$ with two different tasks $\thid$ and $\othid$ in
      $\thidSet$ executing a read or a write on a boolean variable
      with at least one of them writting it (see 
      $\cConfSetRace{n}{p}$ in Fig. \ref{fig:semantics})  of the appendix.
      \item The total
      mapping $\sPhaseMp$ associates $\topOf{\othidSet_\phid,\phid}$
      to each phaser $\phid$, where $\othidSet$ is some subset of $\thidSet$.
      Intuitively, for each phaser $\phid$, we consider all
      registration possibilities (some subset
      $\othidSet\subseteq\thidSet$) while imposing the weakest
      possible constraints on the signal and wait phases of the
      registred tasks. Observe $\sPhaseMp$ is free.
      \end{itemize}
    \item Add to $\sConfSetAssert{n}{p}$ all constraints
      $\tuple{\thidSet,\phidSet,\boolMp,\pcMp,\varMp,\sPhaseMp}$
      where:
      \begin{itemize}
      \item $\tuple{\thidSet,\phidSet,\boolMp,\pcMp,\varMp}$ is a tuple of 
      $\mathcal{C}$ with some task $\thid$ in $\thidSet$ executing an
      assertion on a boolean condition that evaluates to false with
      $\boolMp$ (see $\cConfSetAssert{n}{p}$ in
      Fig. \ref{fig:semantics})  of the appendix.
      \item The total mapping $\sPhaseMp$ associates
      $\topOf{\othidSet_\phid,\phid}$ to each phaser $\phid$, where
        $\othidSet$ is some subset of $\thidSet$. Observe $\sPhaseMp$ is free.
        \end{itemize}
    \item Add to $\sConfSetRuntime{n}{p}$ all constraints
      $\tuple{\thidSet,\phidSet,\boolMp,\pcMp,\varMp,\sPhaseMp}$ where:
      \begin{itemize}
      \item $\tuple{\thidSet,\phidSet,\boolMp,\pcMp,\varMp}$ is a tuple of
      $\mathcal{C}$ with some task $\thid$ in $\thidSet$ is executing
        a statement that involves a phaser variable $\phvar$.
      \item 
      Again, the
      total mapping $\sPhaseMp$ associates
      $\topOf{\othidSet_\phid,\phid}$ to each phaser $\phid$, where
      $\othidSet$ is some subset of $\thidSet$.  In addition, we
      require that either $\varMp(\thid)(\phvar)\undefined$ or
      $\phid=\varMp(\thid)(\phvar)$ with
      $\thid\not\in\registred{\phid}{\sPhaseMp(\phid)}$.  (see
      $\cConfSetRuntime{n}{p}$ in Fig. \ref{fig:semantics})  of the appendix. Observe
      $\sPhaseMp(\ophid)$ is free for all $\ophid$ on which
      $\sPhaseMp$ is defined. Observe $\sPhaseMp$ is free.
      \end{itemize}
    \item Add to $\sConfSetDeadlock{n}{p}$ all constraints
      $\tuple{\thidSet,\phidSet,\boolMp,\pcMp,\varMp,\cPhaseMp}$ where:
      \begin{itemize}
      \item $\tuple{\thidSet,\phidSet,\boolMp,\pcMp,\varMp}$ is a tuple of
      $\mathcal{C}$ where a set of tasks $\thid_0, ..., \thid_{m-1}$
      in $\thidSet$ are executing wait commands on phaser variables
      $\phvar_0, ..., \phvar_{m-1}$ in $\phVarSet$.
      \item Again, the total
      mapping $\sPhaseMp$ associates $\topOf{\othidSet_\phid,\phid}$
      to each phaser $\phid$, where $\othidSet$ is some subset of
      $\thidSet$.  We however require that: $\phid_i=\varMp(\phvar_i)$
      for each $i:0\leq i < m$ and, $\thid_i$ is waiting for
      $\thid_{(i+1)\%m}$ in $\sPhaseMp{\phid_i}$, i.e.,
      $\sPhaseMp{\phid_i}$ imposes the wait phase of $\thid_i$ on
      $\phid_i$ is equal to the signal phase of $\thid_{(i+1)\%m}$ on
      $\phid_i$. In other words, the edge
      $\gedge{\sigVar{\thid_{(i+1)\%m}}{\phid_i}}{0}{\waitVar{\thid_i}{\phid_i}}$ in
      $\sPhaseMp{\phid_i}$.  (see 
      $\cConfSetDeadlock{n}{p}$ in Fig. \ref{fig:semantics})  of the appendix.
      Observe the graphs in $\sPhaseMp$ are not all free since some of them put an
      upper bound on how large some signal value are compared to some wait values.
      They are however of degree 0 since the only negative weights are $-\infty$.
      \end{itemize}
    \end{itemize}
    By construction, we have considered all possible tuples
    $\tuple{\thidSet,\phidSet,\boolMp,\pcMp,\varMp}$ and registration combinations.
    We have all considerd the weakest possible
    constraints on the phases for the registred tasks.
    In addition, any configuration in the denotation constraints will
    belong to the corresponding bad set.
  \end{proof}

  \subsection{Entailment.}
  We say that a constraint 
  $\sConf=\tuple{\thidSet,\phidSet,\boolMp,\pcMp,\varMp,\sPhaseMp}$ is weaker than
  a constraint 
  $\sConf'=\tuple{\thidSet',\phidSet',\boolMp',\pcMp',\varMp',\sPhaseMp'}$, written
  $\sConf \entailedBy \sConf'$,
  to mean the following. First, the two constraints
  have the same number of phasers and tasks, agree on the values of the boolean
  variables and, up to renamings, on the values 
  of the phaser variables and on which tasks are registered to which phasers.
  Second, the constraints on the wait and signal values are stronger
  in $\sConf'$ than in $\sConf$.
  %
  %
  More formally, $\sConf \entailedBy \sConf'$ if $\boolMp=\boolMp'$ and 
  there are  bijections  $\taskMapping:\thidSet\to\thidSet'$ and
  $\phaserMapping: \phidSet\to\phidSet'$ s.t. for each
  $\thid\in\thidSet$ and $\phid\in\phidSet$ the following four conditions hold:
  (i)   $\pcMp(\thid)=\pcMp'(\taskMapping(\thid))$; and (ii)
  $\phaserMapping(\varMp(\thid)(\phvar))=\varMp'(\taskMapping(\thid))(\phvar)$;
  and (iii) $\phaserMapping(\registred{\phid}{\sPhaseMp(\phid)})=\registred{\phaserMapping(\phid)}{\sPhaseMp'(\phaserMapping(\phid))}$; and (iv) 
  $\sPhaseMp(\phid)\entailed\substSet{\sPhaseMp'(\phaserMapping(\phid))}{\setcomp{\waitVar{\taskMapping(\thid)}{\phaserMapping(\phid)}/\waitVar{\thid}{\phid},\sigVar{\taskMapping(\thid)}{\phaserMapping(\phid)}/\sigVar{\thid}{\phid}}{\thid\in\registred{\phid}{\sPhaseMp(\phid)}}}$.
  Clearly, $\sConf \entailedBy \sConf'$ implies 
  $\denotationOf{\sConf'}\subseteq\denotationOf{\sConf}$. We say
  that $\entailedBy$ is sound.

    We can show that $\entailedBy$ is a well-quasi-order\footnote{A
      reflexive and transitive binary relation $\preceq$ is a
      well-quasi-order over a set $A$ if there is no infinite sequence
      $a_0, a_1, \ldots$ of $A$ elements s.t. $a_i\not\preceq a_j$ for
      all $i < j$.} over constraints of bounded degrees and involving
    fixed numbers of tasks and phasers since 
    $\entailed$ is itself
    a well-quasi-ordering over graphs of bounded degrees over a finite
    set of variables (\cite{TCS93:Revesz,FI16:Mayr:Totzke}).

  \begin{lemma}[WQO]
    Given $k,n,p \in\nats$, the entailment relation $\entailedBy$ over
    the set of constraints of degree $k$ involving at most $n$ tasks and
    $p$ phasers is a well-quasi-order.
  \end{lemma}

      \begin{proof}
      Assume an infinite sequence $\sConf_1,\sConf_2,\ldots$ of
      constraints where
      $\sConf_i=\tuple{\thidSet_i,\phidSet_i,\boolMp_i,\pcMp_i,\varMp_i,\sPhaseMp_i}$
      where $|\thidSet_i|=n$ and $|\phidSet_i|=p$ for all $i \geq 1$
      and where all appearing graphs have degree $k$ or less.
      We can assume wlog that $\thidSet_i=\thidSet_j$ and $\phidSet_i=\phidSet_j$
      for any $i,j\geq 1$.
      We show there are $i < j$ such that $\sConf_i \entailedBy \sConf_j$.
      There is a finite number of different values for $\boolMp_i, \pcMp_i, \varMp_i$
      and for the domains of $\sPhaseMp_i(\phid)$ for each $\phid\in\phidSet$.
      We can therefore extract an infinite subsequence where:
      \begin{itemize}
      \item $\boolMp_{i}=\boolMp_{j}$,
        $\pcMp_{i}=\pcMp_{j}$, $\varMp_{i}=\varMp_{j}$, and       
      \item for each $\phid\in\phidSet$, the domains of $\sPhaseMp_{i}(\phid)$
        and $\sPhaseMp_{j}(\phid)$ coincide.
      \end{itemize}
      Observe that for each phaser $\phid\in\phidSet$, the graphs $\sPhaseMp_{i}(\phid)$
      are the same up to possible differences on the weights.
      Let $\preceq$ be the component-wise ordering on vectors.
      Each graph $\sPhaseMp_{i}(\phid)$ is of degree $k$ or less.
      We can therefore organize its 
      weights as a vector $\vectorOf{\sPhaseMp_{i}(\phid)}$ with elements in
      $\ints\cup\set{-\infty,+\infty}$ but where the only allowed negative
      elements are those larger than $-k$.
      The vecotrs can be organized in a way that
      $\vectorOf{\sPhaseMp_{i}(\phid)} \preceq \vectorOf{\sPhaseMp_{j}(\phid)}$ means
      $\sPhaseMp_{i}(\phid)\entailed\sPhaseMp_{j}(\phid)$.
      We then repeat the
      following steps for each $\phid\in\phidSet$:
      using Higman's lemma \cite{higman1952ordering} and the degree boundedness of
        the constraints,
        we extract an infinite sequence of constraints where
        $\vectorOf{\sPhaseMp_{a_b}(\phid)}\preceq\vectorOf{\sPhaseMp_{a_c}(\phid)}$ if
        $a_b < a_c$.        
    \end{proof}

      \begin{figure*}[t]
  \begin{center}
    \scriptsize
    $$
    \begin{array}{cc}
      \multicolumn{2}{c}{
       \infer[\mathtt{(newPhaser~I)}]
             {\tuple{\thidSet,\phidSet,\boolMp,\pcMp,\varMp,\sPhaseMp} \preReducesTo{\thid}{}\tuple{\thidSet,\phidSet\setminus\set{\phid},\boolMp,\pcMp',\varMp',\sPhaseMp'}}
             { \begin{array}{c}
                 \pcMp'=\mapSubst{\pcMp}{\thid}{\phvar:=\newp;\pcMp(\thid)}
                 ~\wedge \\
                 \varMp(\thid)(\phvar)=\phid
                 ~\wedge~ \varMp'=\mapSubst{\varMp}{\thid}{\mapSubst{\varMp(\thid)}{\phvar}{\undefined}}
                 ~\wedge \\
                 \set{\waitVar{\thid}{\phid}\mapsto
                   0,\sigVar{\thid}{\phid}\mapsto 0} \models \sPhaseMp(\phid)
                 ~\wedge~ \sPhaseMp'= \sPhaseMp\setminus\set{\phid} ~\wedge \\
                 (\isRegistred{\othid}{\phid}{\sPhaseMp(\phid)} \implies \othid=\thid)      \\
               \end{array}
             }
      }
            \\
            \\
                         \multicolumn{2}{c}{
             \infer[\mathtt{(newPhaser~II)}]
                   {\tuple{\thidSet,\phidSet,\boolMp,\pcMp,\varMp,\sPhaseMp} \preReducesTo{\thid}{}\tuple{\thidSet,\phidSet\setminus\set{\phid},\boolMp,\pcMp',\varMp',\sPhaseMp'}}
                   { \begin{array}{c}
                       \pcMp'=\mapSubst{\pcMp}{\thid}{\phvar:=\newp;\pcMp(\thid)}
                       ~\wedge \\
                       \varMp(\thid)(\phvar)=\phid
                       ~\wedge~ \varMp'=\mapSubst{\varMp}{\thid}{\mapSubst{\varMp(\thid)}{\phvar}{\ophid}}
                       ~\wedge \\
                       \set{\waitVar{\thid}{\phid}\mapsto
                         0,\sigVar{\thid}{\phid}\mapsto 0} \models \sPhaseMp(\phid)
                       ~\wedge~ \sPhaseMp'= \sPhaseMp\setminus\set{\phid} ~\wedge \\
                       (\isRegistred{\othid}{\phid}{\sPhaseMp(\phid)} \implies \othid=\thid)      \\
                     \end{array}
                   }
                   }
            \\
            \\
            \multicolumn{2}{c}{
        \infer[\mathtt{(signal)}]
              {\sConfTuple \preReducesTo{\thid}{} \tuple{\thidSet,\phidSet,\boolMp,\pcMp',\varMp,\sPhaseMp'}}
              {
                \begin{array}{c}
                  \pcMp'=\mapSubst{\pcMp}{\thid}{\sig{\phvar};\pcMp(\thid)}
                  ~\wedge~ \varMp(\thid)(\phvar)=\phid
                  ~\wedge~ \isRegistred{\thid}{\phid}{\sPhaseMp(\phid)}
                  ~\wedge~ \\ \ggraph
                  = \left( \sPhaseMp(\phid) \intersect \ggraphOf{\wedge_{\othid\in\registred{\phid}{\sPhaseMp(\phid)}}
                    ({\sigVar{\thid}{\phid}> \waitVar{\othid}{\phid} \geq
                      0})} \right) \\ ~\wedge~ \isSat{\ggraph}
                  ~\wedge ~ \sPhaseMp'=\mapSubst{\sPhaseMp}{\phid}
                  {\left(\left(\subst{\ggraph}{\sigVar{\thid}{\phid}}
                    {\sigVar{}{}} \intersect\ggraphOf{\sigVar{\thid}{\phid}=\sigVar{}{}-1}\right)\removeVars\set{\sigVar{}{}}\right)}
                \end{array}
              }
      }
      \\ \\
      \multicolumn{2}{c} {
        \infer[\mathtt{(wait)}]
              {\sConfTuple \preReducesTo{\thid}{} \tuple{\thidSet,\phidSet,\boolMp,\pcMp',\varMp,\sPhaseMp'}}
              {
                \begin{array}{c}
                  \pcMp'=\mapSubst{\pcMp}{\thid}{\wait{\phvar};\pcMp(\thid)}
                  ~\wedge~ \varMp(\thid)(\phvar)=\phid
                  ~\wedge~ \isRegistred{\thid}{\phid}{\sPhaseMp(\phid)}
                  ~\wedge~ \\
                  \ggraph = \left(\sPhaseMp(\phid) \intersect \ggraphOf{\wedge_{\set{\othid\in\registred{\phid}{\sPhaseMp(\phid)}}}(\sigVar{\othid}{\phid}\geq\waitVar{\thid}{\phid} > 0)}\right) \\
                  ~\wedge~ \isSat{\ggraph}
                  ~\wedge~
                  \sPhaseMp'=\mapSubst{\sPhaseMp}{\phid}
                                      {\left(\left(\subst{\ggraph}{\waitVar{\thid}{\phid}}{\waitVar{}{}}
                                        \intersect\ggraphOf{\waitVar{\thid}{\phid}=\waitVar{}{}-1}\right)\removeVars\set{\waitVar{}{}}\right)}
                \end{array}
              }
      }
      \\
      \\
      \multicolumn{2}{c}{ \infer[\mathtt{(drop)}]
      {\sConfTuple \preReducesTo{\thid}{} \tuple{\thidSet,\phidSet,\boolMp,\pcMp',\varMp,\sPhaseMp'}
      }
      {\begin{array}{c} \pcMp'=\mapSubst{\pcMp}{\thid}{\dereg{\phvar};\pcMp(\thid)}
      ~\wedge~ \varMp(\thid)(\phvar)=\phid
      ~\wedge~ \neg\isRegistred{\thid}{\phid}{\sPhaseMp(\phid)}
      ~\wedge \\ \sPhaseMp'=\mapSubst{\sPhaseMp}{\phid}
          {\left(\sPhaseMp(\phid)\intersect
            \ggraphOf{
              (\sigVar{\thid}{\phid}\geq \waitVar{\thid}{\phid} \geq 0)
              \wedge_{\othid\in\registred{\phid}{\sPhaseMp(\phid)}}
              ({\sigVar{\othid}{\phid}\geq \waitVar{\thid}{\phid} \geq 0})
              \wedge
              ({\sigVar{\thid}{\phid}\geq \waitVar{\othid}{\phid} \geq 0})}
            \right)} \end{array}}
      }
      \\
      \\
             \multicolumn{2}{c}{
         \infer[\mathtt{(asynch)}]
               {\sConfTuple \preReducesTo{\thid}{}
                 \tuple{\thidSet\setminus\set{\othid},\phidSet,\boolMp,\pcMp'\setminus\set{\othid},\varMp',\sPhaseMp_n}}
               {
                 \begin{array}{c}
                   \pcMp'=\mapSubst{\pcMp}{\thid}{\asynch{\task}{\phvar_1,\ldots \phvar_k}\{\seqVal_1\};\pcMp(\thid)} ~\wedge~ 
                   \parametersOf{\task}=(\ophvar_1,\ldots \ophvar_k) ~\wedge~ 
                   \othid\in\thidSet\setminus\set{\thid} ~\wedge~ \\ 
                   \varMp'=\varMp\setminus\set{\othid} ~\wedge~ 
                   \pcMp(\othid)=\seqVal_1 ~\wedge~ 
                   \forall i:1\leq i \leq k.~ \varMp(\thid)(\phvar_i)=\varMp(\othid)(\ophvar_i)=\phid_i ~\wedge~ \\
                   (\isRegistred{\thid}{\phid_i}{\sPhaseMp(\phid_i)} 
                   \Leftrightarrow \isRegistred{\othid}{\phid_i}{\sPhaseMp(\phid_i)})
                   ~\wedge~ \\ \ggraph_i
                   = \left(\sPhaseMp(\phid_i)\intersect\ggraphOf{\waitVar{\thid}{\phid_i}=\waitVar{\othid}{\phid_i}\wedge\sigVar{\thid}{\phid_i}=\sigVar{\othid}{\phid_i}}\right)
                   ~\wedge~ \\
                   \isSat{\ggraph_i} ~\wedge~ \sPhaseMp_0=\sPhaseMp 
                   ~\wedge~ \sPhaseMp_i=\mapSubst{\sPhaseMp_{i-1}}{\phid_i}{\left(\ggraph_i\removeVars\set{\waitVar{\othid}{\phid_i},\sigVar{\othid}{\phid_i}}\right)} 
                 \end{array}
               }
       }
      \\
      \\
       \multicolumn{2}{c}
                   {
                     \infer[\mathtt{(exit)}]
                           {\tuple{\thidSet,\phidSet,\boolMp,\pcMp,\varMp,\sPhaseMp} \preReducesTo{\thid}{}
                             \tuple{\thidSet\cup\set{\thid},\phidSet,\boolMp,\pcMp',\varMp',\sPhaseMp'}}
                           {
                             \begin{array}{c}
                               \thid\not\in\thidSet ~\wedge~ \pcMp'=\mapSubst{\pcMp}{\thid}{\exit} ~\wedge~
                               f \in\pFunctionsOf{\phVarSet}{\phidSet} ~\wedge 
                               \varMp'=\varMp\uplus\set{\thid\mapsto f} ~\wedge~ 
                               \ophidSet \subseteq \phidSet ~\wedge~ \\
                               \sPhaseMp'=\substSet{\sPhaseMp}{\setcomp{\phid \leftarrow \cPhaseMp(\phid)\intersect
            \ggraphOf{
              (\sigVar{\thid}{\phid}\geq \waitVar{\thid}{\phid} \geq 0)
              \wedge_{\othid\in\registred{\phid}{\sPhaseMp(\phid)}}
              ({\sigVar{\othid}{\phid}\geq \waitVar{\thid}{\phid} \geq 0})
              \wedge
              ({\sigVar{\thid}{\phid}\geq \waitVar{\othid}{\phid} \geq 0})}
         }{~\phid\in\ophidSet}}
                             \end{array}
                           }
                   }            
    \end{array}
    $$
  \caption{Derivation rules for computing $\ipre{\thid}{\sConf}{}$ as
    union of all 
    $\smsetcomp{\sConf'}{{\sConf}\preReducesTo{\thid}{\sConf'} \textrm{ with } \sConf=\sConfTuple \textrm{ and } \thid\in\thidSet}$.}
    \label{fig:pre:symbolic}
    \end{center}
\end{figure*}

  \section{Verification Procedure}
  \label{sec:proc}
  

  \begin{procedure}
    \caption{check($\program$,$\sConfSetBad$,$\thNum$, $\phNum$), a simple working list procedure for checking constraints reachability.}
    \label{proc:check}
    \SetKwFunction{KwContinue}{continue} \KwIn{A program
      $\program=\programTuple$, a set $\sConfSetBad$ of pairwise
      $\entailedBy$-incomparable constraints, maximum upper bounds
      $\thNum$ and $\phNum$ (in $\nats\cup\set{+\infty}$) on coexisting
      tasks and phasers.}  \KwOut{A symbolic run to $\sConfSetBad$ or
      the value \texttt{unreachable}} Initialize both
    $\sConfSetWorking$ and $\sConfSetVisited$ to
    $\setcomp{\tuple{\sConf,\sConf}}{\sConf\in\sConfSetBad}$\;
    \While{there exists $\tuple{\sConf,\trace}\in\sConfSetWorking$}{
      remove $\tuple{\sConf,\trace}$ from $\sConfSetWorking$\; let
      $\sConfTuple=\sConf$\; \lIf{$|\thidSet|>\thNum$ or
        $|\phidSet|>\phNum$}{\KwContinue}
      \lIf{$\cConfInit\models\sConf$}{\Return $\trace$}
      \ForEach{$\thid\in\thidSet$}{
        \ForEach{$\sConf'\in\pre{\thid}{\sConf}$}{
          \If{$\osConf\not\entailedBy\sConf'$ for all $\tuple{\osConf,\_}\in\sConfSetVisited$}{
            Remove from $\sConfSetWorking$ and $\sConfSetVisited$ each $\tuple{\osConf,\_}$
            for which $\sConf'\entailedBy\osConf$\;
            Add $\tuple{\sConf', \sConf'\cdot\thid\cdot\trace}$ to both $\sConfSetWorking$ and $\sConfSetVisited$\;
          }
        }
       }
    }
    \Return {unreachable} \;
  \end{procedure}

We discuss in the following the procedure $\ref{proc:check}$ depicted below and
assume a program $\program$ and a set $\sConfSetBad$ of
constraints the reachability of which we want to check.
$\sConfSetBad$ can for example be any subset of
$\sConfSetDeadlock{n}{p}$ (degree 0) or of $\sConfSetAssert{n}{p}$
(free) in case we want to check the possibility
of a deadlock or of an assertion violation.

It is not difficult to show that
$\denotationOf{\ipre{\thid}{\sConf}{}}$ (obtained as described in
Fig.(\ref{fig:pre:symbolic}) of the appendix) coincides with
$\smsetcomp{\cConf'}{\cConf'\reducesTo{\thid}{}\cConf \textrm{ and }
  \cConf\in\denotationOf{\sConf}}$. Using the soundness of
$\entailedBy$, we can show by induction the partial correctness of the
procedure
\textup{check(}$\program$\textup{,}$\sConfSetBad$\textup{,}$+\infty$\textup{,}$+\infty$\textup{)}.

\begin{lemma}[Partial correctness]
\label{lem:correctness}
If \textup{check(}$\program$\textup{,}$\sConfSetBad$\textup{,}$+\infty$\textup{,}$+\infty$\textup{)}
returns \texttt{unreachable}, then
$\cConfInit\not\reducesTo{*}{}\denotationOf{\sConfSetBad}$.
If it returns a trace
$\sConf_n\cdot\thid_n\cdots\thid_1\cdot\sConf_1$ 
then there are $\cConf_n,\ldots\cConf_1$ with
$\cConf_n=\cConfInit$, $\cConf_1\in\denotationOf{\sConfSetBad}$
and 
$\cConf_i\reducesTo{\thid_i}{}\cConf_{i-1}$ for $i:1<i\leq n$.
\end{lemma}
    \begin{proof}
      Sketch. Assume a constraint $\sConf=\sConfTuple$.
      It should be clear that we can
      complete the definition of $\preReducesTo{}{}$ (with the rules for
      the non-phaser-related statements) in such a way
      that $\denotationOf{\cup_{\thid\in\thidSet}\ipre{\thid}{\sConf}{}}$ coincides with
       $\smsetcomp{\cConf'}{{\cConf'}\reducesTo{}{}{\cConf} \textrm{
          with } \cConf\in\denotationOf{\sConf}}$.
      We can establish by induction the following procedure invariants. At iteration $m$:
      \begin{itemize}
      \item $\sConfSetWorking_m\subseteq\sConfSetVisited_m$
      \item $\tuple{\sConf_n,
        \sConf_n\cdot\thid_n\cdots\thid_1\cdot\sConf_1}\in\sConfSetVisited_m$
        then there are $\cConf_n,\ldots\cConf_1$ with
        $\cConf_n\in\denotationOf{\sConf_n}$ and
        $\cConf_1\in\denotationOf{\sConfSetBad}$ and
        $\cConf_i\reducesTo{\thid_i}{}\cConf_{i-1}$ for $i:1<i\leq n$.
      \item If $\cConf=\cConf_m$ is such that:
        $\cConf_i\reducesTo{\thid_i}{}\cConf_{i-1}$ for $i:1<i\leq m$
        with $\cConf_1\in\denotationOf{\sConfSetBad}$ then, by
        soundness of $\entailedBy$, there is $\tuple{\sConf,
          \tau}\in\sConfSetVisited_m$ such that
        $\cConf_m\in\denotationOf{\sConf}$.
      \end{itemize}
      \end{proof}

\begin{theorem}[Free termination]
\label{thm:free:termination}
\textup{check(}$\program$\textup{,}$\sConfSetBad$\textup{,}$\thNum$\textup{,}$\phNum$\textup{)}
terminates for ${\thNum,\phNum}\in\nats$ and free $\sConfSetBad$.
\end{theorem}
\begin{proof}
Sketch. Freeness is preserved by 
the $\mathtt{pre}$ computation (Fig.(\ref{fig:pre:symbolic})  of the appendix).
Then we use, similar to \cite{coverability,wsts:everywhere},
the fact that non-termination would allow us to build
an infinite sequence of constraints passing the test at line 9 of the procedure.
More concretely, let $\sConf_1, \sConf_2, \ldots$ be the sequence of constraints
that pass (in that order) the test  at line 9.
If a constraint $\sConf_i$ is replaced by another constraint $\sConf_k$, then $\sConf_k\entailedBy\sConf_i$.
This holds for any replaced constraint. By transitivity, this means all replacements happen with weaker
constraints.
This means $\sConf_i \not\entailedBy \sConf_j$, for any $i<j$. This violates well quasi ordering of
$\entailedBy$ over the degree-0 constraints. 
\end{proof}

In order to check reachability of arbitrary constraints, we may need to force termination.
We do this by soundly bounding the degree of generated constraints using a relaxation $\rho_k$.
The relaxation $\rho_k(\sConfTuple)$
    replaces, in each graph $\sPhaseMp(\phid)$,
    each weight $k''$ s.t. $k''<-k$ with $-\infty$.

    \NoCaptionOfAlgo
    \begin{algorithm}
      \ForEach{$\sConf''\in\pre{\thid}{\sConf}$}{
        Let $\sConf'=\relaxDegree{\sConf''}{k}$\;
      }
  \caption{
    Fig.~6. Systematic relaxation}
    \end{algorithm}
%

\begin{theorem}[Forced termination]
\label{thm:forced:termination}
Independently of the degree of
$\sConfSetBad$,
Procedure \textup{check(}$\program$\textup{,}$\sConfSetBad$\textup{,}$\thNum$\textup{,}$\phNum$\textup{)},
where line 8 is replaced by the two
lines of Fig.~(6), is sound and guaranteed to terminate.
\end{theorem}
\begin{proof}
  Soundness is due to the validity of $\relaxDegree{\sConf}{k}\entailedBy\sConf$
  while the termination argument relies, similarly to Theorem~(\ref{thm:free:termination}),
  on well-quasi orederness of $\entailedBy$ on the set of constraints with bounded degree and
  fixed numbers of tasks and phasers.
\end{proof}

\section{The parameterized case}
\label{sec:param}

We apply view abstraction \cite{view:abstraction:vmcai13}
to derive a sound analysis for
programs with
arbitrary numbers of coexisting tasks.
In this preliminary work, only the number of tasks is parameterized.
%
Assume a program $\program=\programTuple$.

\paragraph{Constraint views.}
The view $\restrict{\sConfTuple}{\othidSet}$ of a constraint
$\sConfTuple$ wrt.  set $\othidSet\subseteq\thidSet$ 
is the tuple
$\smtuple{\othidSet,\phidSet,\boolMp,\restrictPc{\pcMp}{\othidSet},
  \restrictPc{\varMp}{\othidSet},
  \smsetcomp{\phid\mapsto\restrictGraph{\sPhaseMp(\phid)}{\othidSet}}{\phid\in\phidSet}}$
where:
\begin{itemize}
\item [i)] $\restrictPc{\pcMp}{\othidSet}$ is the restriction of $\pcMp$ to the domain $\othidSet$.
\item[ii)] $\restrictPc{\varMp}{\othidSet}$ is the restriction of $\varMp$
  to the domain $\othidSet$.  
\item[iii)] $\restrictGraph{\sPhaseMp(\phid)}{\othidSet}$ is obtained
  by restricting the variables of $\sPhaseMp(\phid)$ to $\setcomp{\sigVar{\thid}{\phid},\waitVar{\thid}{\phid}}{\thid\in\othidSet}$, i.e.,
  capturing $\othidSet$'s signals and waits.
\end{itemize}

Intuitively, views are obtained by projecting on subsets of tasks. 
%
A view $\restrict{\sConf}{\othidSet}$ is
said to be of size $k=|\othidSet|$.
%
%
For $k\geq 0$, we aim to compute a least fixpoint $\osConfSet_k$
starting from 
$\alpha_k(\constraintOf{\cConfInit})$
and satisfying 
$\osConfSet_k\entailedBy\alpha_k(post_{\rho}(\gamma_k(\osConfSet_k)))$ ($\sConfSet \entailedBy \osConfSet$
if there is $\sConf\in\sConfSet$ with $\sConf\entailedBy\osConf$ for each $\osConf\in\osConfSet$).
The idea is to restrict the computations to a finite number of tasks 
using the abstraction $\alpha_k$ 
and the concretization $\gamma_k$ (see 
\cite{view:abstraction:vmcai13} for more details). 
The operator $post_{\rho}$ consists in a simple post computation (see Fig. (\ref{fig:spost}) of the appendix)
followed by a relaxation $\rho$ that bounds the degree of the obtained
views to some predetermined value. 
The relaxation is applied to ensure termination for bounded view sizes.
Intuitively, any $post_\rho$ reachable constraint will have all its
views of size $k$ captured by some views in $\osConfSet_k$.
%

Abstraction and conretizations are defined as follows:
\begin{enumerate}
\item Abstraction $\alpha_k(\sConfSet)$ of 
set $\sConfSet$ is the union $\alpha_k^1(\sConfSet)\cup
  \alpha_k^2(\sConfSet)$ where: 
  \begin{enumerate}
    \item
  $\alpha_k^1(\sConfSet)=\setcomp{\sConfTuple\in\sConfSet}{|\thidSet|\leq k}$
  \item
    $\alpha_k^2(\sConfSet)=\setcomp{\restrict{\sConf}{\thidSet}}{\sConf\in\sConfSet\wedge|\thidSet|=k}$.
  \end{enumerate}
\item Concretization $\gamma_k(\sConfSet)$ of a set 
  $\sConfSet$ is the union $\gamma_k^1(\sConfSet)\cup
  \gamma_k^2(\sConfSet)$ where:
  \begin{enumerate}
    \item
  $\gamma_k^1(\sConfSet)=\sConfSet$, and 
  \item
    $\gamma_k^2(\sConfSet)$ is the set of $\sConf=\sConfTuple$ constraints
    s.t. $\othidSet \subseteq \thidSet \wedge |\othidSet|=k$ implies $\restrict{\sConf}{\othidSet}\in\sConfSet$
  \end{enumerate}
\end{enumerate}

We first compute the fixpoint using only $\alpha^1_k$ and $\gamma^1_k$.
If the denotation of the fixpoint's views intersects the one of $\sConfSetBad$,
then we return the faulty trace. Because we only use $\alpha^1_k$ and $\gamma^1_k$,
the only approximation is due to the relaxation $\rho$.
The precision of the relaxation can then be increased.
%
%
Otherwise, we 
use $\alpha_k$ and $\gamma_k$ in order to over-approximate the fixpoint.
%
If the obtained views still do not include bad configurations, then we conclude
that no matter how many tasks are generated, bad configurations are  unreachable.
In case the views do include bad configurations then we restart the analysis with
views of bigger sizes (i.e., restart with a larger $k$).

\begin{figure*}[t]
  \begin{center}
    \scriptsize
    $$
    \begin{array}{cc}
\multicolumn{2}{c}{
      \infer[\mathtt{(newPhaser)}]
            {\sConfTuple \reducesTo{\thid}{}\tuple{\thidSet,\phidSet',\boolMp,\mapSubst{\pcMp}{\thid}{\seqVal},\varMp',\sPhaseMp'}}
            {
              \begin{array}{c}
                \pcMp(\thid)=\phvar:=\newp;\seqVal  ~~\wedge~ \phid\not\in\phidSet ~\wedge~\\
                \phidSet'=\phidSet\cup\set{\phid} ~\wedge~ \varMp'=\mapSubst{\varMp}{\thid}{\mapSubst{\varMp(\thid)}{\phvar}{\phid}} \\
                \wedge~ \sPhaseMp'=\mapSubst{\sPhaseMp}{\phid}{\ggraphOf{\waitVar{\thid}{\phid}=\sigVar{\thid}{\phid}=0}}
              \end{array}
            }
            }
            \\
            \\
            \multicolumn{2}{c}{
       \infer[\mathtt{(signal)}]
             {\sConfTuple \reducesTo{\thid}{}
               \tuple{\thidSet,\phidSet,\boolMp,\mapSubst{\pcMp}{\thid}{\seqVal},\sPhaseMp'} }
             {
               \begin{array}{c}
               \pcMp(\thid)=\sig{\phvar};\seqVal ~ \wedge \\
               \varMp(\thid)(\phvar)=\phid ~\wedge~ 
               \isRegistred{\thid}{\phid}{\sPhaseMp(\phid)} ~\wedge \\
               \sPhaseMp'=\mapSubst{\sPhaseMp}{\phid}
                                   {
                                     \left(\left(\subst{\sPhaseMp(\phid)}{\sigVar{\thid}{\phid}}{\sigVar{}{}}
                                     \intersect\ggraphOf{\sigVar{\thid}{\phid}=\sigVar{}{}+1}\right)\removeVars\set{\sigVar{}{}}\right)
                                   }
               \end{array}
             }
             }
      \\
      \\
      \multicolumn{2}{c}{
      \infer[\left(\!\!\!\!\begin{array}{c}\mathtt{assert.} \\ \mathtt{ok}\end{array}\!\!\!\!\right)]
            {\sConfTuple \reducesTo{\thid}{}
              \tuple{\thidSet,\phidSet,\boolMp,\mapSubst{\pcMp}{\thid}{\seqVal},\varMp,\sPhaseMp}}
            {
              \begin{array}{c}
                \pcMp(\thid)=\assert{\cond};\seqVal ~ \wedge \\
                \boolMp(\cond)=\true
                \end{array}
            }
            }
            \\
            \\
            \multicolumn{2}{c}{
       \infer[\mathtt{(drop)}]
             {\sConfTuple \reducesTo{\thid}{}
               \tuple{\thidSet,\phidSet,\boolMp,\mapSubst{\pcMp}{\thid}{\seqVal},\varMp,\sPhaseMp'} }
             {\begin{array}{c}
                 \pcMp(\thid)=\dereg{\phvar};\seqVal ~ \wedge 
                 \varMp(\thid)(\phvar)=\phid ~\wedge~ \\
                 \isRegistred{\thid}{\phid}{\sPhaseMp(\phid)} ~\wedge~  
                 \sPhaseMp'=\mapSubst{\sPhaseMp}{\phid}
                                     {\sPhaseMp(\phid)\ominus\set{\waitVar{\thid}{\phid},\sigVar{\thid}{\phid}}}
               \end{array}
             }
             }
       \\
       \\
       \multicolumn{2}{c}{
         \infer[\mathtt{(asynch)}]
               {\sConfTuple \reducesTo{\thid}{}
                 \tuple{\thidSet,\phidSet,\boolMp,\mapSubst{\pcMp'}{\thid}{\seqVal_2},\varMp,\sPhaseMp'}}
               {
                 \begin{array}{c}
                   \pcMp(\thid)=\asynch{\task}{\phvar_1,\ldots \phvar_k}\{\seqVal_1\};\seqVal_2 ~\wedge~ 
                   \parametersOf{\task}=(\ophvar_1,\ldots \ophvar_k) ~\wedge~ \\
                   \textrm{for all } i:1\leq i\leq k.~\varMp(\thid)(\phvar_i)=\phid_i ~\wedge~  \isRegistred{\thid}{\phid_i}{\sPhaseMp(\phid_i)} ~\wedge~ \\
                   \othid\not\in\thidSet ~\wedge~ 
                   \varMp'=\mapSubst{\varMp}{\othid}{\setcomp{\ophvar_i\mapsto \varMp(\thid)(\phvar_i)}{1\leq i\leq k}} ~\wedge~ 
                   \pcMp'=\mapSubst{\pcMp}{\othid}{\seqVal_1} ~\wedge~ \\
                   \sPhaseMp'=\substSet{\sPhaseMp}
                                       {\setcomp{\phid_i\leftarrow
                                           \sPhaseMp(\phid_i)\intersect\ggraphOf{\waitVar{\thid}{\phid_i}=\waitVar{\othid}{\phid_i}\wedge\sigVar{\thid}{\phid_i}=\sigVar{\othid}{\phid_i}}}
                                         {\varMp(\thid)(\phvar_i)=\phid_i \textrm{ for } \phid_i\in\phidSet\textrm{ and }1\leq i\leq k}}
                 \end{array}
               }
       }
       \\
       \\
       \multicolumn{2}{c}
                   {
                     \infer[\mathtt{(wait)}]
                           {\sConfTuple \reducesTo{\thid}{}
                             \tuple{\thidSet,\phidSet,\varMp,\mapSubst{\pcMp}{\thid}{\seqVal},\sPhaseMp'} }
                           {\begin{array}{c}
                               \pcMp(\thid)=\wait{\phvar};\seqVal ~\wedge~
                               \varMp(\thid)(\phvar)=\phid ~\wedge~
                               \isRegistred{\thid}{\phid}{\sPhaseMp(\phid)} ~\wedge~  
                               \\
         \pi =\bigwedge_{\othid\in\registred{\phid}{\sPhaseMp(\phid)}}\left(\waitVar{\thid}{\phid}<\sigVar{\othid}{\phid}\right)
         ~\wedge 
         \isSat{\sPhaseMp(\phid)\intersect \pi} ~\wedge~ \\
         \sPhaseMp'= \mapSubst{\sPhaseMp}{\phid}
         {\left(
         \left(\subst{\left(\sPhaseMp(\phid)\intersect \pi\right)}{\waitVar{\thid}{\phid}}{\waitVar{}{}}
         \intersect\ggraphOf{\waitVar{\thid}{\phid}=\waitVar{}{}+1}\right)
         \removeVars\set{\waitVar{}{}}\right)
         }
                           \end{array}}
                   }
       \\
       \\
       \multicolumn{2}{c}
                   {
                     \infer[\mathtt{(exit)}]
                           {\tuple{\thidSet,\phidSet,\boolMp,\pcMp,\varMp,\sPhaseMp} \reducesTo{\thid}{}
                             \tuple{\thidSet\setminus\set{\thid},\phidSet,\boolMp,\pcMp',\varMp',\sPhaseMp'}}
                           {
                             \begin{array}{c}
                               \pcMp(\thid)=\exit ~\wedge~ 
                               \varMp'=\varMp\setminus\set{\thid} ~\wedge~ 
                               \pcMp'=\pcMp\setminus\set{\thid} ~\wedge~ \\
                               \sPhaseMp'=\substSet{\sPhaseMp}{\setcomp{\phid \leftarrow \sPhaseMp(\phid)\ominus\set{\waitVar{\thid}{\phid},\sigVar{\thid}{\phid}}}{~\phid\in\phidSet}}
                             \end{array}
                           }
                   }            
    \end{array}
    $$
  \end{center}

  \caption{Computing
    $\ipost{\thid}{\sConfTuple}{}$ for some $\thid\in\thidSet$.
    Phasers' related part of the symbolic successor computation used during
    the view abstraction based 
    parameterized analysis of phaser programs.  }
  \label{fig:spost}
\end{figure*}

\section{Experimental Results}
\label{sec:experiments}



We report 
on experiments with our open source prototype 
{\em hjVerify}\footnote{\url{https://gitlab.ida.liu.se/apv/hjVerify}}
for the verification of phaser programs.
We conducted experiments on 12 different programs (some of which are
from \cite{Cogumbreiro:dynamicverifphasers:2015}). We considered both
deadlocks and assertions reachability problems. For
each property, we considered correct and buggy versions.
This gave 48 different instances with 2 to 3
phasers and 2 to 4 tasks (except for the parameterized case).
Our tool uses global phaser and task variables
as in \cite{Cogumbreiro:dynamicverifphasers:2015}.
%
%
We have experimented with adapting the view abstraction technique
\cite{view:abstraction:vmcai13} to verify phaser programs generating arbitrary
many tasks, i.e., parameterized verification where the number of phasers is fixed.
(see appendix for more details.)
We report on two parameterized examples. 
Experiments were conducted on a 2.9GHz processor
with 8GB of memory. 
\begin{center}
{\small
  \begin{tabular}{|c|r c c|}
  \hline
  ~\textbf{program}~ & ~\textbf{property}~ & \textbf{safe / buggy} &  \textbf{times} \\
  \hline
  \hline
    \multirow{2}{*}{01.Loopless} &   \multirow{1}{*}{deadlock:}  & {ok / trace}  &    1s / 1s  \\ \cline{2-4} 
                    &   \multirow{1}{*}{assertion:}   & {ok / trace}    &    1s / 1s  \\ \cline{2-4}
  \hline
  \hline
    02.Iterative    &   \multirow{1}{*}{deadlock:}  & {ok / trace}    &    1s / 1s \\ \cline{2-4}
          averaging     &   \multirow{1}{*}{assertion:} & {ok / trace}    &    1s / 1s \\ \cline{2-4}
  \hline
    \hline
      03.Ordered                  &   \multirow{1}{*}{deadlock:}  & {ok / trace}    &    1s / 1s \\ \cline{2-4}
        phasers       &   \multirow{1}{*}{assertion:} & {ok / trace}    &    13s / 1s \\ \cline{2-4}
         \hline
             \hline
         \multirow{2}{*}{04.Conditional}                   &   \multirow{1}{*}{deadlock:}  & {ok / trace}    &    2s / 1s \\ \cline{2-4}
               &   \multirow{1}{*}{assertion:} & {ok / trace}    &    4s / 7s \\ \cline{2-4}
         \hline
                      \hline
         \multirow{2}{*}{05.Loop Synch.}                   &   \multirow{1}{*}{deadlock:}  & {ok / trace}    &    178s / 145s \\ \cline{2-4}
               &   \multirow{1}{*}{assertion:} & {ok / trace}    &    7s / 13s \\ \cline{2-4}
         \hline
                      \hline
         \multirow{2}{*}{06.Nested forks}                  &   \multirow{1}{*}{deadlock:}  & {ok / trace}    &    2s / 1s \\ \cline{2-4}
               &   \multirow{1}{*}{assertion:} & {ok / trace}    &    1s / 1s \\ \cline{2-4}
          \hline
         \hline
        07.Conditional        &   \multirow{1}{*}{deadlock:}  & {ok / trace}    &    1s / 1s \\ \cline{2-4}
         membership      &   \multirow{1}{*}{assertion:} & {ok / trace}    &    12s / 3s \\ \cline{2-4}
         \hline
                      \hline
        08.Producer-                   &   \multirow{1}{*}{deadlock:}  & {ok / trace}    &    37s / 222s \\ \cline{2-4}
        consumer        &   \multirow{1}{*}{assertion:} & {ok / trace}    &    79s / 34s \\ \cline{2-4}
         \hline
                      \hline
        09.Parameterized                   &   \multirow{1}{*}{deadlock:}  & {ok / trace}    &    20s / 1s \\ \cline{2-4}
        loopless        &   \multirow{1}{*}{assertion:} & {ok / trace}    &    67s / 1s \\ \cline{2-4}
         \hline
                      \hline
        10.Parameterized                   &   \multirow{1}{*}{deadlock:}  & {ok / trace}    &    1s / 1s \\ \cline{2-4}
        iterative-averaging      &   \multirow{1}{*}{assertion:} & {ok / trace}    &    1s / 1s \\ \cline{2-4}
         \hline
                      \hline
        \multirow{2}{*}{11.Running-2}                   &   \multirow{1}{*}{deadlock:}  & {ok / trace}    & 5s / 1s    \\ \cline{2-4}
              &   \multirow{1}{*}{assertion:} & {ok / trace}    &  26s / 4s   \\ \cline{2-4}
         \hline
                      \hline
        \multirow{2}{*}{12.Running-3}                   &   \multirow{1}{*}{deadlock:}  & {ok / trace}    & 4318s / 128s    \\ \cline{2-4}
              &   \multirow{1}{*}{assertion:} & {ok / trace}    &   18631s / 54s  \\ \cline{2-4}
         \hline
         \end{tabular}
         }
\end{center}

%
%
Our implemented procedure does not eagerly concretize all tasks states
as described in the predecessor computation of Section \ref{sec:symbolic}.
Instead we collect conditions on the phases of the threads that did
not take any action yet and lazily concretize them.
%
Reported times for 
checking deadlocks 
are the sums of the times required to check reachability for each cycle.
%
%
The prototype is only a proof of concept.  For instance, the example
(12.Running-3) is a variant of (11-Running-2) where a task
instance is spawned twice leading to two symmetrical tasks (out of four).
This required up to three orders of magnitude more time to check. We
believe partial order reduction techniques would help here. Other
relevant heuristics would be to make use of priority queues and to
organize the minimal sets. 
All examples are available on the tool homepage.

\section{Conclusion}
\label{sec:conc}

We have proposed a gap-order based reachability analysis for phaser
programs.
We have showed our analysis to be exact and guaranteed to terminate
when checking runtime, race and assertion errors. 
We have established the undecidability of deadlock verification
and explained how to turn our analysis into a sound
over-approximation.
To the best of our knowledge, this is beyond the capabilities of
current verification techniques which currently only target concrete inputs to phaser programs.
%
We are currently working on tackling the parameterized case and have
obtained preliminary encouraging results.
Apart from improving the scalability of the tool and from using it in
combination with predicate abstraction and abstract interpretation in
order to analyze actual source code, we are investigating the
applicability of the presented techniques for the verification of
similar synchronization constructs.


%


\begin{thebibliography}{10}
\providecommand{\url}[1]{#1}
\csname url@samestyle\endcsname
\providecommand{\newblock}{\relax}
\providecommand{\bibinfo}[2]{#2}
\providecommand{\BIBentrySTDinterwordspacing}{\spaceskip=0pt\relax}
\providecommand{\BIBentryALTinterwordstretchfactor}{4}
\providecommand{\BIBentryALTinterwordspacing}{\spaceskip=\fontdimen2\font plus
\BIBentryALTinterwordstretchfactor\fontdimen3\font minus
  \fontdimen4\font\relax}
\providecommand{\BIBforeignlanguage}[2]{{%
\expandafter\ifx\csname l@#1\endcsname\relax
\typeout{** WARNING: IEEEtran.bst: No hyphenation pattern has been}%
\typeout{** loaded for the language `#1'. Using the pattern for}%
\typeout{** the default language instead.}%
\else
\language=\csname l@#1\endcsname
\fi
#2}}
\providecommand{\BIBdecl}{\relax}
\BIBdecl

\bibitem{X10:2005}
\BIBentryALTinterwordspacing
P.~Charles, C.~Grothoff, V.~Saraswat, C.~Donawa, A.~Kielstra, K.~Ebcioglu,
  C.~von Praun, and V.~Sarkar, ``X10: An object-oriented approach to
  non-uniform cluster computing,'' \emph{SIGPLAN Not.}, vol.~40, no.~10, pp.
  519--538, Oct. 2005. [Online]. Available:
  \url{http://doi.acm.org/10.1145/1103845.1094852}
\BIBentrySTDinterwordspacing

\bibitem{JDVW:phasers:2008}
J.~Shirako, D.~M. Peixotto, V.~Sarkar, and W.~N. Scherer, ``Phasers: a unified
  deadlock-free construct for collective and point-to-point synchronization,''
  in \emph{22nd annual international conference on Supercomputing}.\hskip 1em
  plus 0.5em minus 0.4em\relax ACM, 2008, pp. 277--288.

\bibitem{cave2011habanero}
V.~Cav{\'e}, J.~Zhao, J.~Shirako, and V.~Sarkar, ``Habanero-java: the new
  adventures of old x10,'' in \emph{Proceedings of the 9th International
  Conference on Principles and Practice of Programming in Java}.\hskip 1em plus
  0.5em minus 0.4em\relax ACM, 2011, pp. 51--61.

\bibitem{accumulator:applications:pdp09}
J.~Shirako, D.~M. Peixotto, V.~Sarkar, and W.~N. Scherer, ``Phaser
  accumulators: A new reduction construct for dynamic parallelism,'' in
  \emph{Parallel \& Distributed Processing, 2009. IPDPS 2009. IEEE
  International Symposium on}.\hskip 1em plus 0.5em minus 0.4em\relax IEEE,
  2009, pp. 1--12.

\bibitem{Cogumbreiro:dynamicverifphasers:2015}
\BIBentryALTinterwordspacing
T.~Cogumbreiro, R.~Hu, F.~Martins, and N.~Yoshida, ``Dynamic deadlock
  verification for general barrier synchronisation,'' in \emph{20th ACM SIGPLAN
  Symposium on Principles and Practice of Parallel Programming}, ser. PPoPP
  2015.\hskip 1em plus 0.5em minus 0.4em\relax New York, NY, USA: ACM, 2015,
  pp. 150--160. [Online]. Available:
  \url{http://doi.acm.org/10.1145/2688500.2688519}
\BIBentrySTDinterwordspacing

\bibitem{Le:barriers:permissions:2013}
D.-K. Le, W.-N. Chin, and Y.-M. Teo, ``Verification of static and dynamic
  barrier synchronization using bounded permissions,'' in \emph{International
  Conference on Formal Engineering Methods}.\hskip 1em plus 0.5em minus
  0.4em\relax Springer, 2013, pp. 231--248.

\bibitem{Anderson:jpf:hj:2014}
\BIBentryALTinterwordspacing
P.~Anderson, B.~Chase, and E.~Mercer, ``Jpf verification of habanero java
  programs,'' \emph{SIGSOFT Softw. Eng. Notes}, vol.~39, no.~1, pp. 1--7, Feb.
  2014. [Online]. Available: \url{http://doi.acm.org/10.1145/2557833.2560582}
\BIBentrySTDinterwordspacing

\bibitem{havelund:jpf:2000}
K.~Havelund and T.~Pressburger, ``Model checking java programs using java
  pathfinder,'' \emph{International Journal on Software Tools for Technology
  Transfer (STTT)}, vol.~2, no.~4, pp. 366--381, 2000.

\bibitem{FI16:Mayr:Totzke}
R.~Mayr and P.~Totzke, ``Branching-time model checking gap-order constraint
  systems,'' \emph{Fundamenta Informaticae}, vol. 143, no. 3-4, pp. 339--353,
  2016.

\bibitem{TCS14:BozzelliPinchinat}
\BIBentryALTinterwordspacing
L.~Bozzelli and S.~Pinchinat, ``Verification of gap-order constraint
  abstractions of counter systems,'' \emph{Theoretical Computer Science}, vol.
  523, pp. 1 -- 36, 2014. [Online]. Available:
  \url{http://www.sciencedirect.com/science/article/pii/S030439751300892X}
\BIBentrySTDinterwordspacing

\bibitem{TCS93:Revesz}
P.~Z. Revesz, ``A closed-form evaluation for datalog queries with integer
  (gap)-order constraints,'' \emph{Theoretical Computer Science}, vol. 116,
  no.~1, pp. 117--149, 1993.

\bibitem{LahiriMusuvathi:utvpi:2005}
\BIBentryALTinterwordspacing
S.~Lahiri and M.~Musuvathi, ``An efficient decision procedure for utvpi
  constraints,'' in \emph{Frontiers of Combining Systems (FroCos '05)}.\hskip
  1em plus 0.5em minus 0.4em\relax Springer Verlag, May 2005. [Online].
  Available:
  \url{https://www.microsoft.com/en-us/research/publication/an-efficient-decision-procedure-for-utvpi-constraints/}
\BIBentrySTDinterwordspacing

\bibitem{higman1952ordering}
G.~Higman, ``Ordering by divisibility in abstract algebras,'' \emph{Proceedings
  of the London Mathematical Society}, vol.~3, no.~1, pp. 326--336, 1952.

\bibitem{coverability}
P.~A. Abdulla, K.~Cerans, B.~Jonsson, and Y.-K. Tsay, ``General decidability
  theorems for infinite-state systems,'' in \emph{Logic in Computer Science,
  1996. LICS'96. Proceedings., Eleventh Annual IEEE Symposium on}.\hskip 1em
  plus 0.5em minus 0.4em\relax IEEE, 1996, pp. 313--321.

\bibitem{wsts:everywhere}
\BIBentryALTinterwordspacing
A.~Finkel and P.~Schnoebelen, ``Well-structured transition systems
  everywhere!'' \emph{Theoretical Computer Science}, vol. 256, no.~1, pp. 63 --
  92, 2001. [Online]. Available:
  \url{http://www.sciencedirect.com/science/article/pii/S030439750000102X}
\BIBentrySTDinterwordspacing

\bibitem{view:abstraction:vmcai13}
P.~A. Abdulla, F.~Haziza, and L.~Hol{\'\i}k, ``All for the price of few,'' in
  \emph{International Workshop on Verification, Model Checking, and Abstract
  Interpretation}.\hskip 1em plus 0.5em minus 0.4em\relax Springer, 2013, pp.
  476--495.

\end{thebibliography}









\def\Nst#1{$#1^{st}$}\def\Nnd#1{$#1^{nd}$}\def\Nrd#1{$#1^{rd}$}\def\Nth#1{$#1^{th}$}

\end{document}